\documentclass[12pt,a4paper,twoside]{article}

\usepackage{fancyhdr,graphicx,subfigure}
\usepackage{amsmath,amsfonts,amsthm}
\usepackage[abs]{overpic}
\usepackage{epsfig}

\newcommand{\tr}{{\rm tr}}

\newcommand{\p}{\partial}

\newcounter{note}

\newcounter{notelist}

\newtheorem{theorem}{Theorem}[section]
\newtheorem{lemma}{Lemma}[section]
\newtheorem{corollary}{Corollary}[section]
\newtheorem{remark}{Remark}[section]
\newtheorem{proposition}{Proposition}[section]
\newtheorem{definition}{Definition}[section]

\numberwithin{equation}{section}

\begin{document}

\title{The Riemann-Hilbert approach to double scaling limit of random matrix eigenvalues near the "birth of a cut"
transition}
\author{M. Y. Mo}
\date{}
\maketitle

\begin{abstract}
In this paper we studied the double scaling limit of a random
unitary matrix ensemble near a singular point where a new cut is
emerging from the support of the equilibrium measure. We obtained
the asymptotic of the correlation kernel by using the
Riemann-Hilbert approach. We have shown that the kernel near the
critical point is given by the correlation kernel of a random
unitary matrix ensemble with weight $e^{-x^{2\nu}}$. This provides a
rigorous proof of the previous results in \cite{Ey}.
\end{abstract}
\section{Introduction}
In this paper we studied a double scaling limit of the unitary
random matrix model with the probability distribution
\begin{equation}\label{eq:rm}
Z_{n,N}^{-1}\exp(-N\tr (V(M)))dM,\quad
Z_{n,N}=\int_{\mathcal{H}_n}\exp(-N\tr (V(M)))dM
\end{equation}
defined on the space $\mathcal{H}_n$ of Hermitian $n\times n$
matrices $M$, where $V$ is real analytic and satisfies
\begin{equation*}
\lim_{x\rightarrow\pm\infty}\frac{V(x)}{\log(x^2+1)}=+\infty.
\end{equation*}
The eigenvalues $x_1,\ldots,x_n$ of the matrices in this ensemble is
distributed according to the probability distribution (See, e.g.
\cite{M}, \cite{D})
\begin{equation}\label{eq:proeig}
\mathcal{P}^{(n,N)}(x_1,\ldots,x_n)d^nx=\hat{Z}_{n,N}^{-1}e^{-N\sum_{j=1}^nV(x_i)}\prod_{j<k}(x_j-x_k)^2dx_1\ldots
dx_n,
\end{equation}
where $\hat{Z}_{n,N}$ is the normalization constant.

A particular important object is the $m$-point correlation function
$\mathcal{R}_m^{(n,N)}(x_1,\ldots,x_m)$
\begin{equation}\label{eq:corre}
\mathcal{R}_m^{(n,N)}(x_1,\ldots,x_m)=\frac{n!}{(n-m)!}\int_{\mathbb{R}}\cdots\int_{\mathbb{R}}
P^{(n,N)}(x_1,\ldots,x_n)dx_{m+1}\ldots dx_n.
\end{equation}
The correlation function is a very useful quantity in the
calculation of probabilities. In fact the 1-point correlation
function $\mathcal{R}_1^{(n,N)}(x)$ gives the probability density of
finding an eigenvalue at the point $x$. (Note that, however, the
$m$-point correlation function $\mathcal{R}_m^{(n,N)}$ is not a
probability density in general.)

A well-known result concerning the $m$-point correlation function is
that it admits a determinantal expression with a kernel constructed
from orthogonal polynomials. (See e.g. \cite{Dy}, \cite{M})

To be precise, let $\pi_n(x)$ be the degree $n$ monic orthogonal
polynomials with weight $e^{-NV(x)}$ on $\mathbb{R}$. \cite{Szego}
\begin{equation}\label{eq:op}
\int_{\mathbb{R}}\pi_n(x)\pi_m(x)e^{-NV(x)}dx=h_n\delta_{nm}.
\end{equation}
Let us construct the correlation kernel by
\begin{equation*}
K_{n,N}(x,x^{\prime})=e^{-\frac{1}{2}N(V(x)+V(x^{\prime}))}\sum_{j=0}^{n-1}\frac{\pi_j(x)\pi_j(x^{\prime})}{h_j}.
\end{equation*}
By the Christoffel-Darboux formula, this kernel can be expressed in
terms of the two orthogonal polynomials $\pi_n(x)$ and
$\pi_{n-1}(x)$ instead of the whole sum:
\begin{equation}\label{eq:kernel}
K_{n,N}(x,x^{\prime})=e^{-\frac{1}{2}N(V(x)+V(x^{\prime}))}\frac{\pi_n(x)\pi_{n-1}(x^{\prime})-\pi_n(x^{\prime})
\pi_{n-1}(x)}{h_{n-1}(x-x^{\prime})}
\end{equation}
Then the $m$-point correlation function can be written as the
determinant of the kernel (\ref{eq:kernel}) \cite{Dy}, \cite{M},
\cite{Po}
\begin{equation*}
\mathcal{R}_{m}^{(n,N)}(x_1,\ldots,x_m)=\det\left(K_{n,N}(x_j,x_k)\right)_{1\leq
j,k\leq m}
\end{equation*}

In the limit $n,N\rightarrow\infty$, $\frac{n}{N}\sim 1$, the
1-point correlation function $\mathcal{R}^{(n,N)}_1(x)$ of the
ensemble (\ref{eq:rm}) is asymptotic to the \it equilibrium measure
\rm $\rho(x)$ \cite{D}, \cite{Jo}, \cite{TS}:
\begin{equation*}
\lim_{n,N\rightarrow\infty,\frac{n}{N}\rightarrow1}\mathcal{R}_1^{(n,N)}(x)=\rho(x),
\end{equation*}
where the $\rho(x)dx=d\mu_{min}(x)$ is the density of the unique
measure $\mu_{min}(x)$ that minimizes the energy
\begin{equation*}
I(\mu)=-\int_{\mathbb{R}}\int_{\mathbb{R}}\log|x-y|d\mu(x)d\mu(y)+\int_{\mathbb{R}}V(x)d\mu(x)
\end{equation*}
among all Borel probability measures $\mu$ on $\mathbb{R}$. The fact
that $\mu_{min}(x)$ admits a probability density follows from the
assumption that $V(x)$ is real and analytic \cite{DKP}. Moreover, it
was shown in \cite{DKP} that for real and analytic $V(x)$, the
equilibrium measure is supported on a finite union of intervals.

The following conditions are satisfied by the equilibrium density
$\rho(x)$ \cite{D}, \cite{TS}
\begin{equation}\label{eq:ineq}
\begin{split}
&2\int_{\mathbb{R}}\log |x-s|\rho(s)ds-V(x)= l, \quad
x\in\textrm{Supp}(\rho(x)),\\
&2\int_{\mathbb{R}}\log|x-s|\rho(s)ds-V(x)\leq l, \quad
x\in\mathbb{R}/\textrm{Supp}(\rho(x)).
\end{split}
\end{equation}
For a generic potential $V(x)$, the inequality in (\ref{eq:ineq}) is
satisfied strictly
\begin{equation*}
\begin{split}
&2\int_{\mathbb{R}}\log|x-s|\rho(s)ds-V(x)> l, \quad
x\in\mathbb{R}/\textrm{Supp}(\rho(x)).
\end{split}
\end{equation*}
However, for some special potential $V(x)$, this inequality may not
be strict and we may have
\begin{equation*}
\begin{split}
&2\int_{\mathbb{R}}\log|x-s|\rho(s)ds-V(x)= l, \quad x=x^{\ast}
\end{split}
\end{equation*}
at some point $x^{\ast}\notin\textrm{Supp}(\rho(x))$. In this case,
if we change the potential slightly, a new interval may emerge from
the support of the equilibrium measure. This is the `birth of new
cut' critical limit that we are going to consider in this paper.

According to \cite{DKV}, the `birth of new cut' critical limit is a
singularity of type I for the random matrix model (\ref{eq:rm}).
Other singular cases include:
\begin{enumerate}
\item Type II singularity is the case where the equilibrium density
vanishes at a point $x^{\ast}$ inside the support.

\item Type III singularity is the case where the equilibrium density
vanishes faster than a square-root at an edge point $x^{\ast}$ of
the support. (Generically it vanishes like a square-root at the
edge)
\end{enumerate}

The asymptotic behavior of a random matrix ensemble near singular
points has been studied extensively \cite{BE}, \cite{BI2},
\cite{BK}, \cite{C}, \cite{CK}, \cite{CK2}, \cite{CKV}, \cite{CV},
\cite{DK}, \cite{KI}, \cite{Sc}. In these studies, one considers a
one or multi-parameter family of potential $V_{t_j}(x)$ in which the
singular point is achieved at $t_j=t_j^c$. One then studies the
asymptotic behavior of the random matrix model (\ref{eq:rm}) when
$t_j$ is close to $t_j^c$. The `double scaling limit' is the study
of the these asymptotic behavior when the differences between $t_j$
and $t_j^c$ are coupled with $n$ and $N$. A remarkable feature is
that in the double scaling limit, a universality can be observed.
Upon a suitable scaling of the variables $x$ and $x^{\prime}$, the
asymptotic behavior of the kernel (\ref{eq:kernel}) near the
critical point $x^{\ast}$ depends only on the type of singularity
rather than the potential $V(x)$ itself.

In many cases, the behavior of the kernel in a double scaling limit
is described by integrable hierarchies such as the Painlev\'e
equations. In the case of the type II singularity, \cite{BI},
\cite{CK}, \cite{CKV} and \cite{Sc} has shown that the kernel can be
described by the Hastings-McLeod solution of the Painlev\'e II
equation in the double scaling limit. While for the type III
singularity, the kernel can be described by the Painlev\'e I
transcendent \cite{CV}, \cite{DK}. In \cite{Ey}, the double scaling
limit of the `birth of new cut' was studied and the kernel was
described by the orthogonal polynomials with weight $e^{-x^{2\nu}}$
on the real axis. However, the formulae derived in \cite{Ey} have
not been rigorously proven and it is the purpose of this paper to
provide a rigorous proof of these results.

\subsection{Statement of results}

We should now introduce some notations and state the results in this
paper.

In this paper, we should consider a one parameter family of
potential $V_t(x)=\frac{V(x)}{t}$ parametrised by $t=\frac{n}{N}$.
We should consider the double scaling limit of $t\rightarrow 1$ and
$n$, $N\rightarrow\infty$ such that
\begin{subequations}
\begin{align}
\lim_{n,N\rightarrow\infty}\frac{\log
n}{n}\left(\frac{n}{N}-1\right)&=U_+> 0,\quad n> N\label{eq:scale00}\\
\lim_{n,N\rightarrow\infty}n^{k}\left(\frac{n}{N}-1\right)&=U_-\leq
0,\quad n\leq N, \quad k\in
\left[1-\frac{1}{2\nu},\infty\right)\label{eq:scalep}
\end{align}
\end{subequations}
exist. In particular, for $t\leq 1$, we considered the regime where
$t-1$ is of order $n^{-k}$ for any $k$ greater than or equal to
$1-\frac{1}{2\nu}$, while the scaling for $t>1$ is fixed.

Let us now state the assumptions that are used in this study. Since
the main point of this study is the treatment of the critical point
$x^{\ast}$, we will assume the followings:
\begin{enumerate}
\item The support of the equilibrium density $\rho(x)$ consists of one
interval only, that is, the first equation of (\ref{eq:ineq}) holds
precisely on a single interval $(a,b)$. Without lost of generality,
we will assume that $a=-2$ and $b=2$.

\item The equilibrium measure does not vanish at any interior point
of $(-2,2)$.

\item The point $x^{\ast}$ is the only point outside
$\textrm{Supp}(\rho)(x)$ where the inequality in (\ref{eq:ineq}) is
not strict and we assume that $x^{\ast}>2$.

\item As pointed out in \cite{DKV2}, the function
$2\int_{\mathbb{R}}\log|x-s|\rho(x)ds-V(x)-l$ vanishes to an even
order at $x^{\ast}$. We will assume that this order of vanishing is
$2\nu$.
\end{enumerate}
Let the equilibrium measure of $V_t(x)$ be $\rho^t(x)$ such that
\begin{equation}\label{eq:rhot}
\begin{split}
&2\int_{\mathbb{R}}\log |x-s|\rho^t(s)ds-V_t(x)= l_t, \quad
x\in\textrm{Supp}(\rho^t(x)),\\
&2\int_{\mathbb{R}}\log|x-s|\rho^t(s)ds-V_t(x)\leq l_t, \quad
x\in\mathbb{R}/\textrm{Supp}(\rho^t(x)),
\end{split}
\end{equation}
and denote by $c_{x^{\ast}}$ the following
\begin{equation}\label{eq:cstar}
c_{x^{\ast}}=n\left(\frac{V_t(x^{\ast})}{2}+\frac{l_{t}}{2}-\int_{\mathbb{R}}\rho^t(s)\log|x^{\ast}-s|ds\right)\geq
0.
\end{equation}
It is known that both $t\rho^t(x)$ and the support of $\rho^t(x)$
are increasing with $t$ \cite{KM}, \cite{DK}, \cite{TS}, \cite{T}.
In particular, for $t\leq 1$, the equilibrium measure is supported
on one interval while for $t$ slightly greater than 1, the
equilibrium measure is supported on 2 intervals.

Let $\mathcal{S}_t$ be the support of $\rho^t(x)$. Then in
\cite{BR}, it was shown that the equilibrium measure
$d\mu_t(x)=\rho^t(x)dx$ satisfies the Buyarov-Rakhmanov equation
\begin{equation}\label{eq:br}
\mu_t=\frac{1}{t}\int_0^t\omega_{\mathcal{S}_{\tau}}d\tau,
\end{equation}
where $\omega_{\mathcal{S}_{\tau}}$ is the equilibrium measure of
the set $\mathcal{S}_{\tau}$. Namely, it is the unique probability
measure supported on $\mathcal{S}_{\tau}$ that minimizes the
logarithmic potential
\begin{equation*}
\begin{split}
I(\tilde{\mu})=\int\int -\log|s-t|d\tilde{\mu}(s)d\tilde{\mu}(t)
\end{split}
\end{equation*}
among all the Borel probability measures $\tilde{\mu}$ supported on
$\mathcal{S}_{\tau}$.

If $\mathcal{S}_{\tau}$ consists of one interval only, then
$\omega_{\mathcal{S}_{\tau}}(x)$ is given by
\begin{equation*}
\omega_{\mathcal{S}_{\tau}}=\frac{1}{\pi\sqrt{(b_{\tau}-x)(x-a_{\tau})}}dx,\quad
x\in (a_{\tau},b_{\tau}).
\end{equation*}
In particular, we have, at $t=1$
\begin{equation}\label{eq:br1}
\lim_{t\rightarrow
1}\frac{t\mu_t(x)-\mu(x)}{t-1}=\frac{1}{\pi\sqrt{4-x^2}}dx=w(x)dx.
\end{equation}
The fact that $w(x)dx$ is the equilibrium measure on the interval
$[-2,2]$ means that
\begin{equation}\label{eq:wineq}
\begin{split}
&\int_{-2}^2w(s)\log|x-s|ds=\frac{\varsigma}{2},\quad x\in [-2,2],\\
&\int_{-2}^2w(s)\log(x-s)ds=\log x+O(1),\quad x\rightarrow\infty
\end{split}
\end{equation}
for some constant $\varsigma$.

Let us defined a function $\phi(x)$ that is closely related to
$w(x)dx$.
\begin{equation}\label{eq:phix}
\phi(x)=\frac{\varsigma}{2}+\int_{-2}^2w(s)\log(x^{\ast}-s)ds.
\end{equation}

In this paper, we will use an anzatz in \cite{Ey} to construct an
approximated equilibrium density $\tilde{\rho}^t(x)$ for $t>1$ and
use it to modify the Riemann-Hilbert problem of the orthogonal
polynomials (\ref{eq:op}).

We shall denote the correlation kernel for the random matrix model
\begin{equation}\label{eq:fintierm}
Z_{m,\nu}^{-1}\exp(-\tr (M^{2\nu}))dM,\quad
Z_{m,\nu}=\int_{\mathcal{H}_{m}}\exp(-\tr M^{2\nu})dM
\end{equation}
by $K_{m}^{\nu}(x,x^{\prime})$. That is,
\begin{equation}\label{eq:kfinite}
K_{m}^{\nu}(x,x^{\prime})=e^{-\frac{(x^{\prime})^{2\nu}+x^{2\nu}}{2}}\frac{\pi_{m}^{\nu}(x)\pi_{m-1}^{\nu}(x^{\prime})
-\pi_{m}^{\nu}(x^{\prime})\pi_{m-1}^{\nu}(x)}{h_{m-1}^{\nu}(x-x^{\prime})},
\end{equation}
where $\pi_{m}^{\nu}(x)$ is the degree $m$ monic orthogonal
polynomial on $\mathbb{R}$ with respect to the weight
$e^{-x^{2\nu}}$ and $h_{m}^{\nu}$ is the corresponding normalization
constant as in (\ref{eq:op}).

We can now state our main result.
\begin{theorem}\label{thm:main}
Let $V(x)$ be real and analytic on $\mathbb{R}$ such that
$\lim_{x\rightarrow\pm\infty}\frac{V(x)}{\log (x^2+1)}=+\infty$. Let
$\rho(x)$ be the density of the equilibrium measure of $V(x)$
supported on the interval $[-2,2]$. Then
\begin{equation*}
\rho(x)=\frac{\sqrt{4-x^2}Q(x)(x-x^{\ast})^{2\nu-1}}{2\pi},x\in[-2,2],
\end{equation*}
where $x^{\ast}>2$ and $Q(x)$ is real analytic on $\mathbb{R}$ with
$Q(x^{\ast})>0$.

Let $n$, $N\rightarrow\infty$ such that (\ref{eq:scale00}) and
(\ref{eq:scalep}) hold and let $u$, $\overline{u}$ be the following
\begin{equation}
\begin{split}
u&=2\nu\phi(x^{\ast})U_+\\
\overline{u}&=\left[2\nu\phi(x^{\ast})U_++\frac{1}{2}\right]
\end{split}
\end{equation}
where $\phi(x^{\ast})$ is defined in (\ref{eq:phix}) and $[x]$ is
the greatest integer that is smaller than or equal to $x$.

Let $K_{n,N}$ be the correlation kernel (\ref{eq:kernel}), then for
$u\notin\mathbb{N}+\frac{1}{2}$, the limit of the kernel is given by
\begin{subequations}
\begin{align}
&\lim_{n,N\rightarrow\infty}\frac{1}{\varphi(x^{\ast})n^{\frac{1}{2\nu}}}K_{n,N}\left(x^{\ast}+\frac{z}{\varphi(x^{\ast})n^{\frac{1}{2\nu}}},
x^{\ast}+\frac{z^{\prime}}{\varphi(x^{\ast})n^{\frac{1}{2\nu}}}\right)
=K_{\overline{u}}^{\nu}(z,z^{\prime}),\quad n>N\label{eq:critker},\\
\begin{split}
&\lim_{n,N\rightarrow\infty}e^{c_{x^{\ast}}}K_{n,N}\left(x^{\ast}+\frac{z}{\varphi(x^{\ast})n^{\frac{1}{2\nu}}}
,x^{\ast}+\frac{z^{\prime}}{\varphi(x^{\ast})n^{\frac{1}{2\nu}}}\right)\\
&=e^{-\frac{z^{2\nu}+(z^{\prime})^{2\nu}}{2}}\frac{1}{8\pi}
\left(\frac{1}{x^{\ast}-\beta_t}-\frac{1}{x^{\ast}-\alpha_t}\right)
,\quad n\leq N\label{eq:critker1}.
\end{split}
\end{align}
\end{subequations}
where $K_{\overline{u}}^{\nu}(z,z^{\prime})$ is defined in
(\ref{eq:kfinite}) and $c_{x^{\ast}}$ is defined in (\ref{eq:cstar})
and $\varphi(x^{\ast})$ is given by
\begin{equation*}
\varphi(x^{\ast})=\left(\frac{Q(x^{\ast})\sqrt{(x^{\ast})^2-4}}{2\nu}\right)^{\frac{1}{2\nu}}.
\end{equation*}
\end{theorem}
The result shows that for $u\notin\mathbb{N}+\frac{1}{2}$, the
correlation kernel near $x^{\ast}$ for $t>1$ is given by the
correlation kernel of a finite random matrix ensemble
(\ref{eq:fintierm}) with size $[u+\frac{1}{2}]$. This confirms the
results in \cite{Ey}. When $u$ goes pass a half integer, the size of
the finite random matrix ensemble jumps by 1 and a non-trivial
transition takes place. This is due to the non-uniform converges of
(\ref{eq:critker}) in $u$ when $u$ is close to a half integer. When
$u$ is close to a half integer, error terms that depends on
$K_{\overline{u}\pm 1}^{\nu}$ which are not seen in
(\ref{eq:critker}) become significant and start taking over the
$K_{\overline{u}}^{\nu}$ terms, which results in a jump when $u$
goes pass a half integer.

Note that (\ref{eq:critker1}) implies that the leading order term of
the kernel at $x^{\ast}$ is $e^{c_{x^{\ast}}}$. This leading term
tends to zero when $n$, $N\rightarrow\infty$ unless $t=1$. This is
not surprising as for $t< 1$, there is no eigenvalue near the point
$x^{\ast}$ and the correlation kernel should be vanishing near
$x^{\ast}$ in the limit.

\begin{remark} Claeys \cite{C} has simultaneously and independently
used the Riemann-Hilbert method to study the birth of new cut double
scaling limit. In Claeys \cite{C}, the case when $\nu=1$ was studied
and the Hermite polynomials was used to construct the asymptotic
kernel. Despite the similarity of our work to \cite{C}, a very
different treatment to the equilibrium measure was used in \cite{C}.
In \cite{C}, the equilibrium measure with total mass
$1-2\frac{t-1}{\log n}\phi(x^{\ast})$ was used to construct the
`$g$-function' for the Deift-Zhou steepest decent method when $t>1$.
Whereas in this paper, we approximated the equilibrium measure by
solving the Buyarov-Rakhmanov equation (\ref{eq:br}) up to a certain
order in $t-1$. We then use this approximated measure to construct
the `$g$-function' for the Deift-Zhou steepest decent method. Also
worth remarking is that in \cite{C}, the behavior of the kernel when
$u$ is close to a half integer was studied.
\end{remark}

This paper is organized as follows. In section \ref{se:equil} we
will use the ansatz obtained in \cite{Ey} to construct an
approximated equilibrium density for $t>1$. We then show that
conditions of the type (\ref{eq:ineq}) are satisfied for this
approximated density outside some neighborhoods of the edge points
and the critical point. We then study the error terms in these
conditions.

In section \ref{se:RH} we will apply the Deift-Zhou steepest decent
method to the Riemann-Hilbert problem of the orthogonal polynomials
(\ref{eq:op}). We will use the approximated density to construct a
`$g$-function' and use it to modify the Riemann-Hilbert problem. We
then approximate this modified Riemmann-Hilbert problem by a
Riemann-Hilbert problem that can be solved explicitly and construct
parametrices to solve this approximated Riemann-Hilbert problem.
These parametrices then give us the asymptotics of the orthogonal
polynomials (\ref{eq:op}). These asymptotics will then be used to
derive the asymptotics of the kernel (\ref{eq:kernel}) in section
\ref{se:ker}.

\section{Equilibrium measure}\label{se:equil}

We will now study the behavior of the equilibrium measure
$\rho^t(x)$ (\ref{eq:rhot}) when $t$ is close to 1. Let $t$ be a
real parameter and let us define
\begin{equation*}
V_t(x)=\frac{1}{t}V(x), \quad t>0.
\end{equation*}
Then $V_1(x)=V(x)$. We shall consider the case when $t\leq 1$ and
$t>1$ separately. For $t>1$, we will replace the eigenvalues on the
newborn interval by a point charge. Let the support of the
equilibrium measure $\mathcal{S}_t$ be
\begin{equation}\label{eq:supp}
\begin{split}
\mathcal{S}_t&=[a_t,b_t],\quad t\leq 1\\
\mathcal{S}_t&=[a_t,b_t]\cup[c_t,d_t],\quad t>1.
\end{split}
\end{equation}
Let us define the function $h^t$ by
\begin{equation}\label{eq:g}
h^t(x)=\int_{\mathbb{R}}\log (x-s)d\mu_t(s)
\end{equation}
where the principal branch of the logarithm is taken in the above,
\begin{equation*}
\begin{split}
\log(x-s)&=\log|x-s|+i\arg(x-s)\\
0&<\arg(x-s)<\pi,\quad s\in\mathbb{R},\quad \Im x>0,\\
-\pi&<\arg(x-s)<0,\quad s\in\mathbb{R},\quad \Im x,0.
\end{split}
\end{equation*}
The boundary values of $h^t(x)$ on the real axis are then
\begin{equation*}
h^t_{\pm}(x)=\int_{\mathbb{R}}\log |x-s|d\mu_t(s)\pm\pi
i\int_{a_t}^xd\mu_t(s)
\end{equation*}
In particular, the function $h^t$ is analytic on
$\mathbb{C}/[a_t,\infty)$ and it satisfies the following
\begin{equation}\label{eq:hineq}
\begin{split}
&h_+^t(x)+h_-^t(x)-V_t(x)+l_t=0,\quad x\in [a_t,b_t]\cup[c_t,d_t]\\
&h_+^t(x)+h_-^t(x)-V_t(x)+l_t<0,\quad x\in\mathbb{R}/\left([a_t,b_t]\cup[c_t,d_t]\cup\{x^{\ast}\}\right)\\
&h_+^t(x)-h_-^t(x)=2\pi i\int_{x}^{b_t}d\mu_t(s),\quad x\in
\mathbb{R}\\
&h^t(x)=\log x+O(x^{-1})\quad x\rightarrow\infty
\end{split}
\end{equation}
In \cite{DKP}, it was shown that for a real analytic potential
$V(x)$ on $\mathbb{R}$, the equilibrium measure $d\mu_t(s)$ can be
expressed in terms of the negative part of an analytic function
$q_t(x)$.
\begin{theorem}\cite{DKP}\label{thm:curve}
Let $V(x)$ be real analytic in a neighborhood $\mathcal{V}$ of the
real axis and let $q_t(x)$ be the following function
\begin{equation}\label{eq:qt}
q_t(x)=\left(\frac{V^{\prime}(x)}{2t}\right)^2-\frac{1}{t}\int_{\mathbb{R}}\frac{V^{\prime}(x)-V^{\prime}(y)}{x-y}d\mu_t(y),\quad
x\in\mathcal{V}.
\end{equation}
Then the equilibrium measure has a density $\rho^t(x)$ which can be
written as
\begin{equation*}
\rho^t(x)=\frac{1}{\pi}\sqrt{-q_t^-(x)},
\end{equation*}
where $q_t^-(x)$ is the negative part of $q_t(x)$, that is,
\begin{equation*}
q_t(x)=q_t^+(x)+q_t^-(x),\quad q_t^+(x)\geq 0,\quad q_t^-(x)\leq 0.
\end{equation*}
Moreover, we have the following
\begin{equation}\label{eq:curve}
q_t(x)=\left(\int_{\mathbb{R}}\frac{\rho^t(y)}{y-x}dy+\frac{V^{\prime}(x)}{2t}\right)^2,\quad
x\in\mathcal{V}.
\end{equation}
\end{theorem}
\subsection{Approximated equilibrium measure for $t>1$}
For $t>1$, a new cut in the support of the equilibrium measure is
emerging at $x=x^{\ast}$. We would like to find an approximation to
the equilibrium measure and study its properties.

The Buyarov-Rakhmanov equation (\ref{eq:br}) for the equilibrium
measure is a nonlinear ODE which is difficult to solve. In
\cite{Ey}, an ansatz was used to solve this differential equation up
to some leading order terms in $t-1$. As this ODE becomes singular
at $t=1$, it is difficult to prove rigorously that the solution in
\cite{Ey} does indeed give the equilibrium measure for $t$ slightly
greater than 1.

Instead of showing that the solution obtained in \cite{Ey} gives the
correct equilibrium measure for $t>1$, we would use the ansatz in
\cite{Ey} to construct an approximated density $\tilde{\rho}^t(x)$,
together with a function $\tilde{h}^t(x)$ analogue to the function
$h^t(x)$ defined in (\ref{eq:g}). We will then show that this
approximated density satisfies conditions of the type
(\ref{eq:ineq}) up to a certain order in $t-1$.

First note that the function $q_t(x)$ defined in (\ref{eq:qt}) has
the following form at $t=1$.
\begin{equation}\label{eq:q}
\sqrt{q(x)}=\frac{1}{2}Q(x)(x-x^{\ast})^{2\nu-1}\sqrt{x^2-4},
\end{equation}
where $Q(x)$ is analytic in a neighborhood $\mathcal{V}$ of the real
axis.

We can now define a function $\tilde{q}^t(x)$ analogous to $q_t(x)$.
\begin{definition}\label{de:q}
Let $\delta t=t-1>0$. Then the function $\tilde{q}^t(x)$ is defined
by
\begin{equation}\label{eq:tqt}
\sqrt{\tilde{q}^t(x)}=\frac{\sqrt{(x-\alpha_t)(x-\beta_t)}}{2}\left(Q(x)H_t(x)\sqrt{(x-x^{\ast})^2-4y^2\left(-\frac{\delta
t}{\log \delta t}\right)^{\frac{1}{\nu}}}+\eta(x)\delta t\right),
\end{equation}
where $\alpha_t$ and $\beta_t$ are,
\begin{equation}\label{eq:alp}
\alpha_t=-2+\frac{\delta t}{(2+x^{\ast})^{2\nu-1}Q(-2)},\quad
\beta_t=2-\frac{\delta t}{(x^{\ast}-2)^{2\nu-1}Q(2)}.
\end{equation}
while $H_t(x)$ is a monic polynomial defined by
\begin{equation}\label{eq:H}
\begin{split}
H_t(x)=(x-x^{\ast})^{2\nu-2}\sum_{k=0}^{\nu-1}\frac{(2k)!}{k!k!}y^{2k}(x-x^{\ast})^{-2k}\left(-\frac{\delta
t}{\log \delta t}\right)^{\frac{k}{\nu}},
\end{split}
\end{equation}
The function $\eta(x)$ is defined by
\begin{equation}\label{eq:eta}
\begin{split}
\eta(x)&=\frac{Q(x)(x-x^{\ast})^{2\nu-1}}{2Q(2)(2-x^{\ast})^{2\nu-1}(x-2)}+\frac{Q(x)(x-x^{\ast})^{2\nu-1}}{2Q(-2)(2+x^{\ast})^{2\nu-1}(x+2)}
-\frac{2}{x^2-4},
\end{split}
\end{equation}
and the constant $y$ is defined by
\begin{equation}\label{eq:y}
y=\left(\frac{4\nu^2\phi(x^{\ast})(\nu-1)!\nu!}{Q(x^{\ast})\sqrt{(x^{\ast})^2-4}(2\nu)!}\right)^{\frac{1}{2\nu}}
\end{equation}
and $\phi(x^{\ast})$ is defined in (\ref{eq:phix}).
\end{definition}
\begin{remark}The function $\eta(x)$ is analytic in the neighborhood
$\mathcal{V}$ of the real axis.
\end{remark}

We will now show that the density defined by the function
$\sqrt{\tilde{q}_t(x)}$ satisfies the Buyarov-Rakhmanov equation
outside a fixed neighborhood of $x^{\ast}$.
\begin{proposition}\label{pro:br}
Let $B_{\delta}^{s}$ be the set
\begin{equation*}
B_{\delta}^{s}=\{x|\quad |x-s|\leq\delta\}
\end{equation*}
and let $r_1=-2$, $r_2=2$ and $r_3=x^{\ast}$. Then for sufficiently
small $\delta t$, there exist compact subset
$\mathcal{K}\subset\mathcal{V}$ and $\delta>0$ independent on $t$,
such that the function $\tilde{q}^t(x)$ satisfies
\begin{equation}\label{eq:qbr}
\frac{\sqrt{\tilde{q}^t(x)}-\sqrt{q(x)}}{t-1}=-\frac{1}{\sqrt{x^2-4}}+O\left(\frac{\delta
t}{\log\delta t}\right),\quad
 x\in\mathcal{K}/\left(\bigcup_{j=1}^3B_{\delta}^{r_i}\cup[-2,2]\right)
\end{equation}
uniformly in
$\mathcal{V}/\left(\bigcup_{j=1}^3B_{\delta}^{r_i}\cup[-2,2]\right)$,
where $\delta t=t-1$.
\end{proposition}
\begin{proof} We will expand (\ref{eq:tqt}) in terms of $\delta t$ and
$-\frac{\delta t}{\log\delta t}$. Let us first consider the product
$H_t(x)\sqrt{(x-x^{\ast})^2-4y^2\left(-\frac{\delta t}{\log\delta
t}\right)^{\frac{1}{\nu}}}$. Let $\delta>0$ be fixed. Then for small
enough $\delta t$, the following Taylor series expansion is valid
outside of $B_{\delta}^{x^{\ast}}$.
\begin{equation}\label{eq:series}
\begin{split}
\sqrt{(x-x^{\ast})^2-4y^2\left(-\frac{\delta t}{\log\delta
t}\right)^{\frac{1}{\nu}}}&=\sum_{j=0}^{\infty}\frac{(2j)!}{j!j!(1-2j)}y^{2j}(x-x^{\ast})^{-2j+1}\left(-\frac{\delta
t}{\log\delta t}\right)^{\frac{j}{\nu}}
\end{split}
\end{equation}
Now from the Taylor series expansion of
$\left((x-x^{\ast})^2-4y^2\left(-\frac{\delta t}{\log\delta
t}\right)^{\frac{1}{\nu}}\right)^{-\frac{1}{2}}$,
\begin{equation*}
\left((x-x^{\ast})^2-4y^2\left(-\frac{\delta t}{\log\delta
t}\right)^{\frac{1}{\nu}}\right)^{-\frac{1}{2}}=\sum_{j=0}^{\infty}\frac{(2j)!}{j!j!}y^{2j}(x-x^{\ast})^{-2j-1}\left(-\frac{\delta
t}{\log\delta t}\right)^{\frac{j}{\nu}},
\end{equation*}
we see that (c.f. \cite{Ey})
\begin{equation*}
H_t(x)=\mathrm{Pol}\left((x-x^{\ast})^{2\nu-1}\left((x-x^{\ast})^2-4y^2\left(-\frac{\delta
t}{\log\delta t}\right)^{\frac{1}{\nu}}\right)^{-\frac{1}{2}}\right)
\end{equation*}
where $\mathrm{Pol}(X)$ denotes the polynomial part of $X$.

Therefore we have
\begin{equation}\label{eq:prod}
\begin{split}
&H_t(x)\sqrt{(x-x^{\ast})^2-4y^2\left(-\frac{\delta t}{\log\delta
t}\right)^{\frac{1}{\nu}}}=(x-x^{\ast})^{2\nu-1}\\
&+\sum_{j=0}^{\infty}y^{2j+2\nu}(x-x^{\ast})^{-2j-1}\left(-\frac{\delta
t}{\log\delta t}\right)^{1+\frac{j}{\nu}}L_j,\\
&L_j=\sum_{p=0}^{\nu-1}\frac{(2p)!}{p!p!}\frac{(2(j+\nu-p)!)}{(j+\nu-p)!(j+\nu-p)!(1-2(j+\nu-p))}
\end{split}
\end{equation}
Then, for a small enough $\delta t$, we have, for
$|x-x^{\ast}|>\delta$,
\begin{equation*}
\sum_{j=0}^{\infty}y^{2j+2\nu}(x-x^{\ast})^{-2j-1}\left(-\frac{\delta
t}{\log\delta t}\right)^{1+\frac{j}{\nu}}L_j=O\left(\frac{\delta
t}{(\log\delta t)}\right).
\end{equation*}
This means that, for $x\in\mathcal{K}/ B_{\delta}^{x^\ast}$, we have
\begin{equation}\label{eq:conv}
H_t(x)\sqrt{(x-x^{\ast})^2-4y^2\left(-\frac{\delta t}{\log\delta
t}\right)^{\frac{1}{\nu}}}=(x-x^{\ast})^{2\nu-1}+O\left(\frac{\delta
t}{\log\delta t}\right),\quad x\notin B_{\delta}^{x^{\ast}}
\end{equation}
Now let us look at the terms of order $\delta t$. Again, for small
enough $\delta t$, the following Taylor series expansions are valid
outside $B_{\delta}^{2}\cup B_{\delta}^{-2}$.
\begin{equation}\label{eq:rt}
\begin{split}
\sqrt{x-\alpha_t}&=\sqrt{x+2}\sum_{j=0}^{\infty}\frac{(-1)^j(2j)!}{j!j!(1-2j)4^j}\left(\frac{\Xi(-2)\delta
t}{x+2}\right)^j,\\
\sqrt{x-\beta_t}&=\sqrt{x-2}\sum_{j=0}^{\infty}\frac{(2j)!}{j!j!(1-2j)4^j}\left(\frac{\Xi(2)\delta
t}{x-2}\right)^j,
\end{split}
\end{equation}
where the function $\Xi(x)$ is defined by
\begin{equation*}
\Xi(x)=\frac{1}{(x-x^{\ast})^{2\nu-1}Q(x)}.
\end{equation*}
The identity (\ref{eq:rt}) implies that, for small enough $\delta
t$, we have, for $x\in\mathcal{K}/\left( B_{\delta}^{-2}\cup
B_{\delta}^{2}\right)$,
\begin{equation}\label{eq:conv2}
\begin{split}
\sqrt{(x-\alpha_t)(x-\beta_t)}&=\sqrt{x^2-4}+\delta
t\left(-\frac{\sqrt{x+2}\Xi(2)}{2\sqrt{x-2}}+\frac{\sqrt{x-2}\Xi(-2)}{2\sqrt{x+2}}\right)
+O((\delta t)^2).
\end{split}
\end{equation}
Combining this with (\ref{eq:conv}) and (\ref{eq:eta}), we see that,
outside of $B_{\delta}^{x^{\ast}}$, the limit (\ref{eq:qbr}) is
given by
\begin{equation}\label{eq:lim}
\begin{split}
\frac{\sqrt{\tilde{q}^t(x)}-\sqrt{q(x)}}{t-1}&=\Bigg(\frac{\sqrt{x+2}\Xi(2)}{4\Xi(x)\sqrt{x-2}}-\frac{\sqrt{x+2}\Xi(2)}{4\Xi(x)\sqrt{x-2}}
\\&+\frac{\sqrt{x-2}\Xi(-2)}{4\Xi(x)\sqrt{x+2}}-\frac{\sqrt{x-2}\Xi(-2)}{4\Xi(x)\sqrt{x+2}}
-\frac{1}{\sqrt{x^2-4}}\Bigg)+O\left(\frac{\delta t}{\log\delta
t}\right),
\end{split}
\end{equation}
which is just
\begin{equation*}
\begin{split}
\frac{\sqrt{\tilde{q}^t(x)}-\sqrt{q(x)}}{t-1}=-\frac{1}{\sqrt{x^2-4}}+O\left(\frac{\delta
t}{\log\delta t}\right),\quad x\in\mathcal{K}/\left(
\bigcup_{j=1}^3B_{\delta}^{r_i}\cup[-2,2]\right).
\end{split}
\end{equation*}
This gives the assertion of the proposition.
\end{proof}
Let us now define the approximated equilibrium density to be
\begin{equation}\label{eq:tildrho}
\begin{split}
\tilde{\rho}^t(x)&=\frac{1}{t\pi}\left(\sqrt{-\tilde{q}^t(x)}\right)_+,\quad
x\in
[\alpha_t,\beta_t]\\
\tilde{\rho}^t(x)&=\frac{\sqrt{(x-\alpha_t)(x-\beta_t)}}{2t\pi}Q(x)H_t(x)\sqrt{\sigma_t^2-(x-x^{\ast})^2},
\\x&\in[x^{\ast}-\sigma_t,x^{\ast}+\sigma_t], \quad
\sigma_t^{\pm}=2y\left(-\frac{\delta t}{\log \delta
t}\right)^{\frac{1}{2\nu}},\\
\tilde{\rho}^t(x)&=0,\quad
x\in\mathbb{R}/[\alpha_t,\beta_t]\cup[x^{\ast}-\sigma_t,x^{\ast}+\sigma_t],
\end{split}
\end{equation}
and let $\tilde{h}^t(x)$ be the following
\begin{equation}\label{eq:tildh}
\tilde{h}^t(x)=\int_{\alpha_t}^{\beta_t}\tilde{\rho}^t(s)\log(x-s)ds+
\int_{x^{\ast}-\sigma_t}^{x^{\ast}+\sigma_t}\tilde{\rho}^t(s)\log(x-s)ds.
\end{equation}
Then we have the following analogue of (\ref{eq:br1}) for
$\tilde{h}^t(x)$.
\begin{proposition}\label{pro:hineq}
For sufficiently small $\delta t$, there exists $\delta>0$ such that
the following is satisfied for $\tilde{h}^t(x)$
\begin{equation}\label{eq:tild}
\begin{split}
\tilde{h}^t(x)&=\frac{h(x)}{t}+\frac{\delta
t}{t}\left(\int_{-2}^2w(s)\log(x-s)ds\right)+O\left(\frac{\delta
t\log(x+2)}{\log\delta t}\right),\\
x&\in\mathbb{C}/\bigcup_{j=1}^3B_{\delta}^{r_j}\cup\mathrm{Supp}(\tilde{\rho}^t(x)),
\end{split}
\end{equation}
where $h(x)$ is the following
\begin{equation*}
h(x)=\int_{-2}^2\rho(s)\log(x-s)ds
\end{equation*}
and $w(s)$ is the equilibrium measure of the interval $[-2,2]$
(\ref{eq:br1}).
\end{proposition}
\begin{proof}
Let us first divide the real axis in to different parts
\begin{equation*}
\mathbb{R}=\bigcup_{j=1}^6\mathbb{R}_j
\end{equation*}
where the $\mathbb{R}_j$ are the following intervals, that is,
\begin{equation}\label{eq:Rno}
\begin{split}
\mathbb{R}_1&=\left[\alpha_t,-2-2\left(2+\alpha_t\right)\right],\quad
\mathbb{R}_2=\left[-2-2\left(2-\alpha_t\right),-2+\frac{\delta}{2}\right],\\
\mathbb{R}_3&=\left[-2+\frac{\delta}{2},2-\frac{\delta}{2}\right],\quad
\mathbb{R}_4=\left[2-\frac{\delta}{2},2-2\left(\beta_t-2\right)\right],\\
\mathbb{R}_5&=\left[2-2\left(\beta_t-2\right),\beta_t\right],\quad
\mathbb{R}_6=\left[x-\sigma_t,x+\sigma_t\right].
\end{split}
\end{equation}
Let us now define $\Gamma$ to be the line right above
$\mathbb{R}_3$,
\begin{equation}\label{eq:Gamma}
\Gamma=\left\{x\Bigg|\quad x=u+i\varepsilon,\quad
u\in\left[-2+\frac{\delta}{2},2-\frac{\delta}{2}\right],\quad
\varepsilon\rightarrow 0^{+}\right\}.
\end{equation}
Then we have
\begin{equation}\label{eq:Rgamma}
\begin{split}
\int_{\mathbb{R}_3}\tilde{\rho}^t(s)\log(x-s)ds&=\int_{\Gamma}\frac{\sqrt{-\tilde{q}^t(s)}}{t\pi}\log(x-s)ds,\\
\int_{\mathbb{R}_3}\rho(s)\log(x-s)ds&=\int_{\Gamma}\frac{\sqrt{-q(s)}}{\pi}\log(x-s)ds,\\
\int_{\mathbb{R}_3}w(s)\log(x-s)ds&=\int_{\Gamma}\frac{1}{\pi\sqrt{4-s^2}}\log(x-s)ds.
\end{split}
\end{equation}
Let $\delta>0$ be such that the power series expansion of $Q(x)$ and
$\eta(x)$ around $\pm 2$ are valid inside $B_{\frac{\delta}{4}}^{\pm
2}$.

First let us consider the integral on $\mathbb{R}_1$. On
$\mathbb{R}_1$, the following power series expansions are valid.
\begin{equation}\label{eq:ser2}
\begin{split}
&Q(s)=\sum_{j=0}^{\infty}Q_{(j,-2)}(s+2)^j,\quad
\eta(s)=\sum_{j=0}^{\infty}\eta_{j}(s+2)^j\\
&\left(\sqrt{\beta_t-s}\right)_+=
\sqrt{\beta_t+2}\sum_{j=0}^{\infty}\lambda_j(s+2)^j,\\
\log(x-s)&=\log(x+2)-\sum_{j=1}^{\infty}\frac{1}{j}\left(\frac{s+2}{x+2}\right)^j,
\end{split}
\end{equation}
where the branch of $\log(x+2)$ is chosen to be the principal
branch.

It is not difficult to check that the coefficients in the above
series remain finite as $\delta t\rightarrow 0$. Moreover, from
(\ref{eq:conv}), we have
\begin{equation*}
H_t(x)\sqrt{(x-x^{\ast})^2-4y^2\left(-\frac{\delta t}{\log\delta
t}\right)^{\frac{1}{\nu}}}=(x-x^{\ast})^{2\nu-1}+O\left(\frac{\delta
t}{-\log\delta t}\right),\quad x\in \mathbb{R}_1.
\end{equation*}
In particular, this means that on $\mathbb{R}_1$ the functions have
the following estimates
\begin{equation*}
\begin{split}
&Q(s)=Q(-2)+O(\delta t),\quad \eta(s)=\eta(-2)+O(\delta t),\\
&\left(\sqrt{\beta_t-s}\right)_+=\sqrt{\beta_t+2}+O(\delta t),\quad
\log(x-s)=\log(x+2)+O(\delta t).
\\
&H_t(s)\sqrt{(s-x^{\ast})^2-4y^2\left(-\frac{\delta t}{\log\delta
t}\right)^{\frac{1}{\nu}}}=-(2+x^{\ast})^{2\nu-1}+O\left(\frac{\delta
t}{\log\delta t}\right).
\end{split}
\end{equation*}
Therefore the integral on $\mathbb{R}_1$ can be evaluated as
\begin{equation}\label{eq:1stcon}
\begin{split}
\int_{\mathbb{R}_1}\tilde{\rho}^t(s)\log(x-s)ds&=
\frac{\sqrt{\beta_t+2}\log(x+2)}{2\Xi(-2)\pi}\int_{\mathbb{R}_1}
\sqrt{s-\alpha_t}ds\left(1+O\left(\delta t\right)\right)\\
&=\frac{\sqrt{\beta_t+2}\log(x+2)(\Xi(-2))^{\frac{1}{2}}}{3\pi}\left(3\delta
t\right)^{\frac{3}{2}}\left(1+O\left(\delta t\right)\right).
\end{split}
\end{equation}
Similarly, the following integrals for $\rho(x)$ and the equilibrium
measure on $[-2,2]$ are given by
\begin{equation}\label{eq:1equilcon}
\begin{split}
\int_{-2}^{-2-2(2+\alpha_t)}\frac{\rho(s)}{t}\log(x-s)ds&=\frac{2^{\frac{5}{2}}\log(x+2)\Xi(-2)^{\frac{1}{2}}}{3t\pi}(\delta
t)^{\frac{3}{2}}\left(1+O(\delta t)\right),\\
\int_{-2}^{-2-2(2+\alpha_t)}\frac{w(s)\log(x-s)}{t}ds&=\frac{\log(x+2)}{t\pi}\sqrt{2\Xi(-2)\delta
t}\left(1+O(\delta t)\right).
\end{split}
\end{equation}
Therefore we have
\begin{equation}\label{eq:order1}
\begin{split}
\int_{\mathbb{R}_1}\tilde{\rho}^t(s)\log(x-s)&-\int_{-2}^{-2-2(2+\alpha_t)}\left(\frac{\rho(s)}{t}-\frac{\delta
t}{t}w(s)\right)\log(x-s)ds\\
&=O\left((\delta t)^{\frac{3}{2}}\log(x+2)\right).
\end{split}
\end{equation}
Next let us consider the integral on $\mathbb{R}_2$. Since
$|x+2|>\delta$, for $s\in \mathbb{R}_2$, we can find constants
independent on $t$ and $s$ such that
\begin{equation*}
\begin{split}
&|Q(s)|<M_Q,\quad |\sqrt{s+2}|<|\sqrt{s-\alpha_t}|<M_{\alpha},\quad
|\eta(s)|<M_{\eta},\\
&|s-x^{\ast}|^{2\nu-1}<M_{x^{\ast}},\quad
|\sqrt{2-s}|<|\sqrt{\beta_t-s}|<M_{\beta},\quad
\left|\frac{\Xi(2)}{s-2}\right|<M_1,\\
&\left|\sqrt{\beta_t-s}-\sqrt{2-s}+\frac{\Xi(2)}{\sqrt{2-s}}\delta
t\right|<M_2(\delta
t)^2,\\
&|H_t(s)-(s-x^{\ast})^{2\nu-1}|<M_H\left(-\frac{\delta t}{\log\delta
t}\right)\\
&\log|x-s|<M_3\log|x+2|.
\end{split}
\end{equation*}
Then, by using the the Taylor series expansion of
$\sqrt{s-\alpha_t}$ in (\ref{eq:rt}), we see that
\begin{equation}\label{eq:ord}
\begin{split}
 \Bigg|t\tilde{\rho}^t(s)&-\rho(s)+\delta
t\frac{1}{\pi\sqrt{4-s^2}}\Bigg|\log|x-s|\leq
\Bigg(E_1\sum_{j=2}^{\infty}\frac{\sqrt{s+2}(2j)!}{j!j!(2j-1)4^j}\left(\frac{|\Xi(-2)|\delta
t}{s+2}\right)^{j}\\
&+E_2\delta
t\sum_{j=1}^{\infty}\frac{\sqrt{s+2}(2j)!}{j!j!(2j-1)4^j}\left(\frac{|\Xi(-2)|\delta
t}{s+2}\right)^j\\
&+E_3\sum_{j=0}^{\infty}\frac{\sqrt{s+2}(2j)!}{j!j!(2j-1)4^j}\left(\frac{|\Xi(-2)|\delta
t}{s+2}\right)^j\left(\frac{-\delta t}{\log\delta
t}\right)\Bigg)\log|x+2|,
\end{split}
\end{equation}
for some positive constants $E_1$, $E_2$ and $E_3$. One needs to be
careful about the terms that contains negative power of $x+2$ as
they may become large in $\mathbb{R}_2$. If we integrate
(\ref{eq:ord}) and consider the leading order term in $\delta t$, we
see that
\begin{equation*}
\Bigg|\int_{\mathbb{R}_2}\left(t\tilde{\rho}^t(s)-\rho(s)+\delta t
w(s)\right)\log|x-s|ds\Bigg|\leq E\frac{\delta t}{-\log\delta
t}\log|x+2|.
\end{equation*}
for some positive constant $E$.

This implies that
\begin{equation*}
\int_{\mathbb{R}_2}\left(t\tilde{\rho}^t(s)-\rho(s)+\delta
tw(s)\right)\log|x-s|ds=O\left(\frac{\delta t\log(x+2)}{\log\delta
t}\right).
\end{equation*}
We then see that
\begin{equation}\label{eq:order2}
\int_{\mathbb{R}_2}\left(\tilde{\rho}^t(s)-\frac{\rho(s)}{t}+\frac{\delta
t}{t}w(s)\right)\log|x-s|ds=O\left(\frac{\delta
t\log(x+2)}{\log\delta t}\right)
\end{equation}
To compute the integral on $\mathbb{R}_3$, observe that for small
enough $\delta t$, the relation (\ref{eq:qbr}) holds uniformly on
$\Gamma$. Therefore by (\ref{eq:Rgamma}), the integral on
$\mathbb{R}_3$ is given by
\begin{equation}\label{eq:3term}
\int_{\mathbb{R}_3}\left(\tilde{\rho}^t(s)-\frac{\rho(s)}{t}+
\frac{\delta t}{t}w(s)\right)\log|x-s|ds=O\left(\frac{\delta
t\log(x+2)}{\log\delta t}\right).
\end{equation}
By applying the argument used for $\mathbb{R}_1$ and $\mathbb{R}_2$
to $\mathbb{R}_4$ and $\mathbb{R}_5$, we obtain
\begin{equation}\label{eq:order4}
\begin{split}
\int_{\mathbb{R}_j}\left(\tilde{\rho}^t(s)-\frac{\rho(s)}{t}+\frac{\delta
t}{t}w(s)\right)\log|x-s|ds=O\left(\frac{\delta
t\log(x+2)}{\log\delta t}\right),\quad j=4,5.
\end{split}
\end{equation}
Let us now consider the contribution from the interval
$[x-\sigma_t,x+\sigma_t]$. From the power series expansions on
$\mathbb{R}_6$, we have the following estimates,
\begin{equation}\label{eq:est}
\begin{split}
&Q(s)=Q(x^{\ast})+O\left(\left(\frac{\delta t}{\log\delta
t}\right)^{\frac{1}{2\nu}}\right),\quad
\eta(s)=\eta(x^{\ast})+O\left(\left(\frac{\delta t}{\log\delta
t}\right)^{\frac{1}{2\nu}}\right)\\
&\sqrt{(s-\alpha_t)(s-\beta_t)}=\sqrt{(x^{\ast})^2-4}+O\left(\left(\frac{\delta
t}{\log\delta
t}\right)^{\frac{1}{2\nu}}\right),\\
&\log(x-s)=\log(x-x^{\ast})+O\left(\left(\frac{\delta t}{\log\delta
t}\right)^{\frac{1}{2\nu}}\right).
\end{split}
\end{equation}
Therefore, the integral on $\mathbb{R}_6$ satisfies the following
estimate.
\begin{equation}\label{eq:3int}
\begin{split}
\int_{\mathbb{R}_6}\tilde{\rho}^t(s)\log(x-s)ds&=\frac{Q(x^{\ast})\sqrt{(x^{\ast})^2-4}\log(x-x^{\ast})}{2\pi}\\
&\times\int_{\mathbb{R}_6}H_t(s)\sqrt{\sigma_t^2-(s-x^{\ast})^2}ds\left(1+O\left(\left(\frac{\delta
t}{\log\delta t}\right)^{\frac{1}{2\nu}}\right)\right).
\end{split}
\end{equation}
To evaluate the integral on the right, let us note that $H_t(x)$ can
be written in the following form \cite{Ey},
\begin{equation}\label{eq:Hform}
H_t(z)=\left(-\frac{\delta t}{\log\delta
t}\right)^{1-\frac{1}{\nu}}P\left((z-x^{\ast})\left(-\frac{\delta
t}{\log\delta t}\right)^{-\frac{1}{2\nu}}\right),
\end{equation}
where $P(s)$ is the following polynomial of degree $2\nu-2$,
\begin{equation}\label{eq:Ppol}
P(s)=\sum_{j=0}^{\nu-1}\frac{(2j)!}{j!j!}y^{2j}s^{2(\nu-1-j)}.
\end{equation}
Then by a change of variable
\begin{equation*}
\xi=(s-x^{\ast})\left(-\frac{\delta t}{\log\delta
t}\right)^{-\frac{1}{2\nu}}
\end{equation*}
in the integral on the right hand side of (\ref{eq:3int}), we have
\begin{equation}\label{eq:integral}
\begin{split}
\int_{\mathbb{R}_6}\tilde{\rho}^t(s)\log(x-s)ds=-\frac{\delta
t}{\log\delta
t}\frac{Q(x^{\ast})\sqrt{(x^{\ast})^2-4}\log(x-x^{\ast})}{2\pi}\\
\times\int_{-2y}^{2y}P(\xi)\left(\sqrt{4y^2-\xi^2}\right)_+d\xi\left(1+O\left(\left(\frac{\delta
t}{\log\delta t}\right)^{\frac{1}{2\nu}}\right)\right).
\end{split}
\end{equation}
To evaluate this integral, we will use the following differential
equation for $P(\xi)$ in \cite{Ey}.
\begin{equation*}
(2\nu-2)P(\xi)-\xi
P^{\prime}(\xi)=\frac{4y^2}{\xi^2-4y^2}\left(P(\xi)-P(2y)\right).
\end{equation*}
Using this and integration by parts, we find that the integral in
(\ref{eq:integral}) is given by
\begin{equation}\label{eq:filling}
\int_{-2y}^{2y}P(\xi)\sqrt{4y^2-\xi^2}d\xi=\frac{2\pi
y^2P(2y)}{\nu}.
\end{equation}
Hence the integral (\ref{eq:integral}) is
\begin{equation*}
\begin{split}
\int_{\mathbb{R}_6}\tilde{\rho}^t(s)\log(x-s)ds&=-\frac{\delta
t}{\log\delta
t}\frac{y^2P(2y)Q(x^{\ast})\sqrt{(x^{\ast})^2-4}\log(x-x^{\ast})}{\nu}\\
&\times\left(1+O\left(\left(\frac{\delta t}{\log\delta
t}\right)^{\frac{1}{2\nu}}\right)\right).
\end{split}
\end{equation*}
Now by the use of induction, one can compute $P(2y)$ easily
\cite{Ey},
\begin{equation*}
P(2y)=(2y)^{2\nu-2}\sum_{j=0}^{\nu-1}\frac{(2k)!}{k!k!}4^{-k}=y^{2\nu-2}\frac{(2\nu)!}{2(\nu-1)!\nu!}
\end{equation*}
This, together with the expression (\ref{eq:y}) for $y$ implies that
\begin{equation}\label{eq:integral1}
\int_{\mathbb{R}_6}\tilde{\rho}^t(s)\log(x-s)ds=-\frac{\delta
t}{\log\delta
t}2\nu\phi(x^{\ast})\log(x-x^{\ast})\left(1+O\left(\left(\frac{\delta
t}{\log\delta t}\right)^{\frac{1}{2\nu}}\right)\right),
\end{equation}
which is of order $\frac{\delta t}{\log\delta t}$. That is,
\begin{equation}\label{eq:order6}
\begin{split}
\int_{\mathbb{R}_6}\tilde{\rho}^t(s)\log(x-s)ds=O\left(\frac{\delta
t\log(x+2)}{\log\delta t}\right).
\end{split}
\end{equation}
Now by adding (\ref{eq:order1}), (\ref{eq:order2}),
(\ref{eq:3term}), (\ref{eq:order4}), (\ref{eq:order6}), we arrive at
(\ref{eq:tild}).
\end{proof}
Now from (\ref{eq:wineq}) and (\ref{eq:tild}), we see that
conditions of the type (\ref{eq:hineq}) are satisfied for
$\tilde{h}^t(x)$.
\begin{corollary}\label{cor:hineq}
For sufficiently small $\delta t$, there exist $\delta>0$ such that
$[x^{\ast}-\sigma_t,x^{\ast}+\sigma_t]\subset B_{\delta}^{x^\ast}$
and that
\begin{equation}\label{eq:thineq}
\begin{split}
&\tilde{h}^t_+(x)+\tilde{h}^t_-(x)-\frac{V(x)}{t}-\frac{\tilde{l}}{t}=\upsilon_t^h(x)\left(\frac{\delta
t}{\log\delta t}\right),\quad x\in[\alpha_t,\beta_t]/\left(B_{\delta}^{-2}\cup B_{\delta}^{2}\right)\\
&\Re\left(\tilde{h}^t_+(x)+\tilde{h}^t_-(x)-\frac{V(x)}{t}-\frac{\tilde{l}}{t}\right)<0,\quad
x\in\mathbb{R}/\left(\bigcup_{j=1}^3B_{\delta}^{r_i}\cup[\alpha_t,\beta_t]\right),\\
&\tilde{h}^t(x)=\left(1+\iota_t\left(\frac{\delta t}{\log\delta
t}\right)\right)\log x+O(1),\quad x\rightarrow\infty.
\end{split}
\end{equation}
where $\tilde{l}$ is the constant $l+(\delta t)\varsigma$. The
function $\upsilon_t^h(x)$ remains uniformly bounded in
$[\alpha_t,\beta_t]$ as $\delta t\rightarrow 0$, while the constant
$\iota_t$ remains finite in the limit. That is, if
\begin{equation*}
\lim_{\delta t\rightarrow 0}\upsilon_t^h(x)=\upsilon^h(x),\quad
\lim_{\delta t\rightarrow 0}\iota_t=\iota,
\end{equation*}
then $\upsilon^h(x)$ is uniformly bounded in $[\alpha_t,\beta_t]$
and $\iota^h$ is finite.
\end{corollary}
Corollary \ref{cor:hineq} suggests that $\tilde{\rho}^t(x)$ is a
good approximation to the actual equilibrium density $\rho^t(x)$.

The following corollary follows immediately from the proof of
proposition \ref{pro:hineq} and the fact that
$h(x^{\ast})-\frac{V(x^{\ast})}{2}-\frac{l}{2}=0$.
\begin{corollary}\label{cor:ge}
Inside $B_{\delta}^{x^{\ast}}$, the following is satisfied.
\begin{equation}\label{eq:ge0}
\int_{\alpha_t}^{\beta_t}\tilde{\rho}^t(s)\log(x-s)ds-\frac{h(x)}{t}=\delta
t\int_{-2}^2w(s)\log(x-s)ds+O\left(\frac{\delta t}{\log\delta
t}\right).
\end{equation}
In particular, by using
$h(x^{\ast})=\frac{V(x^{\ast})}{2}+\frac{l}{2}$, we see that the
following is satisfied at $x=x^{\ast}$.
\begin{equation}\label{eq:ge}
\int_{\alpha_t}^{\beta_t}\tilde{\rho}^t(s)\log(x^{\ast}-s)ds-\frac{V(x^{\ast})}{2t}-\frac{\tilde{l}}{2t}=
\delta t\phi(x^{\ast})+O\left(\frac{\delta t}{\log\delta t}\right),
\end{equation}
where $\phi(x^{\ast})$ is defined in (\ref{eq:phix}).
\end{corollary}
This corollary is essential for the construction of the local
parametrix inside the neighborhood $B_{\delta}^{x^{\ast}}$. (See
section \ref{se:lpara})

\section{Riemann-Hilbert analysis}\label{se:RH}

A result by Fokas, Its and Kitaev \cite{FI} shows that the
orthogonal polynomials (\ref{eq:op}) can be expressed in terms of a
Riemann-Hilbert problem. In this section we will apply the
Deift-Zhou steepest decent method to approximate this
Riemann-Hilbert problem by a Riemann-Hilbert problem that is
solvable explicitly. We will achieve this by using the approximated
equilibrium measure constructed in section \ref{se:equil}. We will
modify the measure $\tilde{\rho}^t(x)dx$ by replacing the charges on
$[x^{\ast}-\sigma_t,x^{\ast}+\sigma_t]$ by a point charge. This then
allows us to construct local parametrix near the critical point
$x^{\ast}$ from orthogonal polynomials with weight $e^{-x^{2\nu}}$
on the real axis.

\subsection{Riemann-Hilbert problem for the orthogonal polynomials}

One important property of the orthogonal polynomials (\ref{eq:op})
is that they can be represented as a solution to a Riemann-Hilbert
problem \cite{FI}.

Consider the following Riemann-Hilbert problem for a matrix-valued
function $Y(x)=Y_{n,N}(x)$.
\begin{equation}\label{eq:RHP}
\begin{split}
&1. \quad \text{$Y(x)$ is analytic on $\mathbb{C}/\mathbb{R}$}\\
&2.\quad Y_+(x)=Y_-(x)\begin{pmatrix} 1 & e^{-NV(x)}  \\ 0 & 1
\end{pmatrix},\quad x\in\mathbb{R}\\
&3. \quad Y(x)=\left(I+O(x^{-1})\right)\begin{pmatrix} x^n & 0
\\ 0 & x^{-n} \end{pmatrix},\quad x\rightarrow\infty
\end{split}
\end{equation}
where $Y_+(x)$ and $Y_-(x)$ denotes the limiting values of $Y(x)$ as
it approaches the left and right hand sides of the real axis. This
Riemann-Hilbert problem has the following unique solution.
\begin{equation*}
Y(x)=\begin{pmatrix} \pi_n(x) & \frac{1}{2\pi i}\int_{\mathbb{R}}\frac{\pi_n(s)e^{-NV(s)}}{s-x}ds  \\
\kappa_{n-1}\pi_{n-1}(x) & \frac{\kappa_{n-1}}{2\pi
i}\int_{\mathbb{R}}\frac{\pi_{n-1}(s)e^{-NV(s)}}{s-x}ds
\end{pmatrix}
\end{equation*}
where $\kappa_{n-1}=-2\pi ih_{n-1}^{-1}$ \cite{D}. The correlation
kernel (\ref{eq:kernel}) can be expressed in terms of the solution
of the Riemann-Hilbert problem $Y(x)$ via \cite{CK}
\begin{equation}\label{eq:corker}
K_{n,N}(x,y)=\frac{e^{-N\frac{V(x)+V(y)}{2}}}{2\pi
i(x-y)}\left(0\quad 1\right)Y_+^{-1}(y)Y_+(x)\begin{pmatrix} 1 \\ 0
\end{pmatrix}
\end{equation}
We shall apply the Deift-Zhou steepest decent method to the
Riemann-Hilbert problem (\ref{eq:RHP}) to obtain the asymptotics for
the orthogonal polynomials and the correlation kernel.

\subsection{Initial transformation of the Riemann-Hilbert problem}
We shall perform a series of transformation to the Riemann-Hilbert
problem (\ref{eq:RHP}) and approximate it with a Riemann-Hilbert
problem that can be solved explicitly. We will than use the solution
of the final model Riemann-Hilbert problem to compute the
asymptotics of the orthogonal polynomials and the correlation
kernel.

\subsubsection{The $g$-function}

To begin with, let us denote $t$ by $t=\frac{n}{N}$ and rewrite the
jump matrix in (\ref{eq:RHP}) as
\begin{equation*}
\begin{pmatrix} 1 & e^{-NV(x)}  \\ 0 & 1
\end{pmatrix}=\begin{pmatrix} 1 & e^{-nV_t(x)}  \\ 0 & 1
\end{pmatrix}
\end{equation*}
where $V_t(x)=\frac{1}{t}V(x)$. We will now define a function
$g^t(x)$ from the function $\tilde{h}^t(x)$ constructed in section
\ref{se:equil} to transform this Riemann-Hilbert problem.

Let $u^t$ be the following
\begin{equation*}
u^t=n\int_{\mathbb{R}_6}\tilde{\rho}^t(s)ds
\end{equation*}
where $\mathbb{R}_6$ is defined in (\ref{eq:Rno}).

For later convenience, let us denote by $\overline{u}^t$ the
non-negative integer closest to $u^t$:
\begin{equation}\label{eq:ubar}
\begin{split}
\overline{u}^t&=\left[u^t+\frac{1}{2}\right],\quad u^t\geq 0\\
\overline{u}^t&=0,\quad u^t <0.
\end{split}
\end{equation}
From (\ref{eq:integral1}), we see that if $u^t>0$, then
\begin{equation}\label{eq:ut}
u^t=-n\frac{\delta t}{\log\delta
t}2\nu\phi(x^{\ast})\left(1+O\left(\left(\frac{\delta t}{-\log\delta
t}\right)^{\frac{1}{2\nu}}\right)\right).
\end{equation}
By inserting the scaling (\ref{eq:scale00}) into (\ref{eq:ut}), we
see that $u^t$ is finite in this limit.

As mentioned before, we would like to replace the charges in the
interval $[x^{\ast}-\sigma_t,x^{\ast}+\sigma_t]$ by a point charge
when $t>1$. We should therefore define the $g$-function to be the
following.
\begin{equation}\label{eq:gfun}
\begin{split}
g^t(x)&=\int_{\alpha_t}^{\beta_t}\left(\tilde{\rho}^t(s)
-\iota_t\frac{\delta t}{\log\delta
t}\frac{8\pi\sqrt{(s-\alpha_t)(\beta_t-s)}}{(\alpha_t+\beta_t)^2}\right)\log(x-s)ds\\
&+\frac{u^t}{n}\log(x-x^{\ast}),\quad
t>1,\\
g^t(x)&=\int_{\alpha_t}^{\beta_t}\rho^t(s)\log(x-s)ds,\quad t\leq1,
\end{split}
\end{equation}
where in the above equations, we have extended the definitions of
the end points $\alpha_t$ and $\beta_t$ (\ref{eq:alp}) to include
the the values of $t$ that are less than or equal to 1:
\begin{equation}\label{eq:alp1}
\begin{split}
\alpha_t&=-2+\frac{\delta t}{(2+x^{\ast})^{2\nu-1}Q(-2)},\quad
\beta_t=2-\frac{\delta t}{(x^{\ast}-2)^{2\nu-1}Q(2)},\quad t>1,\\
\alpha_t&=a_t,\quad \beta_t=b_t,\quad t\leq1.
\end{split}
\end{equation}
Then, from (\ref{eq:ut}) and (\ref{eq:order6}), we see that, for
$x\in\mathbb{C}/B_{\delta}^{x^{\ast}}$,
\begin{equation*}
\begin{split}
\frac{u^t}{n}\log(x-x^{\ast})&=O\left(\frac{\delta
t\log(x-x^{\ast})}{\log\delta
t}\right)\\
\int_{x^{\ast}-\sigma_t}^{x^{\ast}+\sigma_t}\tilde{\rho}^t(s)\log(x-s)ds&=O\left(\frac{\delta
t\log(x-x^{\ast})}{\log\delta t}\right).
\end{split}
\end{equation*}
It then follows from corollary \ref{cor:hineq} and the properties
(\ref{eq:thineq}) that the function $g^t(x)$ satisfies the following
\begin{proposition}\label{pro:gineq}
For sufficiently small $\delta t$, there exist $\delta>0$ such that
$[x^{\ast}-\sigma_t,x^{\ast}+\sigma_t]\subset B_{\delta}^{x^{\ast}}$
and that
\begin{equation}\label{eq:gineq}
\begin{split}
&g^t_+(x)+g^t_-(x)-\frac{V(x)}{t}-\frac{\tilde{l}}{t}=\upsilon_t(x)\left(\frac{\delta
t}{\log\delta t}\right),\quad x\in[\alpha_t,\beta_t]/\left(B_{\delta}^{-2}\cup B_{\delta}^{2}\right)\\
&g^t_+(x)+g^t_-(x)-\frac{V(x)}{t}-\frac{\tilde{l}}{t}<0,\quad
x\in\mathbb{R}/\left(\bigcup_{j=1}^3B_{\delta}^{r_i}\cup[\alpha_t,\beta_t]\right),\\
&g^t(x)=\log x+O(1),\quad x\rightarrow\infty.
\end{split}
\end{equation}
where $\tilde{l}$ is the following constant
\begin{equation*}
\tilde{l}=\left\{
            \begin{array}{ll}
              l-(\delta t)\varsigma, & \hbox{$t>1$;} \\
              t l_t, & \hbox{$t\leq1$.}
            \end{array}
          \right.
\end{equation*}
The function $\upsilon_t(x)$ remains uniformly bounded in
$[\alpha_t,\beta_t]$ as $\delta t\rightarrow 0$. That is, if
\begin{equation*}
\lim_{\delta t\rightarrow 0}\upsilon_t(x)=\upsilon(x),
\end{equation*}
then $\upsilon(x)$ is uniformly bounded in $[\alpha_t,\beta_t]$. In
particular, $\upsilon_t(x)$ is zero when $t\leq 1$.
\end{proposition}

The function $g^t(x)$ is analytic on $\mathbb{C}/(-\infty,x^{\ast})$
and has the following jump discontinuities on $(-\infty,\alpha_t)$
and $(\beta_t,x^{\ast})$.
\begin{equation}\label{eq:gd}
\begin{split}
g^t_+(x)&-g^t_-(x)=2\pi i\int_{x}^{\beta_t}\left(\tilde{\rho}^t(s)
-\iota_t\frac{\delta t}{\log\delta
t}\frac{8\pi\sqrt{(s-\alpha_t)(\beta_t-s)}}{(\alpha_t+\beta_t)^2}\right)ds\\
 g^t_+(x)&-g^t_-(x)=2\pi i,\quad x\in
(-\infty,\alpha_t)\\
g^t_+(x)&-g^t_-(x)=2\pi i\frac{u^t}{n},\quad x\in (\beta_t,
x^{\ast}).
\end{split}
\end{equation}
we can now transform the Riemann-Hilbert problem with the function
$g^t(x)$.

Let $T(x)$ be the following function
\begin{equation}\label{eq:Tx}
T(x)=e^{\frac{-n\tilde{l}\sigma_3}{2t}}Y(x)e^{-ng^t(x)\sigma_3}e^{\frac{n\tilde{l}\sigma_3}{2t}},
\end{equation}
where $\sigma_3$ is the Pauli matrix
\begin{equation*}
\sigma_3=\begin{pmatrix} 1 &0  \\ 0 & -1
\end{pmatrix}.
\end{equation*}

Then $T(x)$ is a solution to the following Riemann-Hilbert problem.
\begin{equation*}
\begin{split}
&1. \quad \text{$T(x)$ is analytic in $\mathbb{C}/\mathbb{R}$};\\
&2. \quad T_+(x)=T_-(x)J_T(x),\quad x\in\mathbb{R};\\
&3. \quad T(x)=I+O(x^{-1}),\quad x\rightarrow\infty.
\end{split}
\end{equation*}
where $J_T(x)$ is the following matrix on $\mathbb{R}$.
\begin{equation*}
J_T(x)=\begin{pmatrix} e^{-n(g_+^t(x)-g_-^t(x))} & e^{n\left(g_+^t(x)+g_-^t(x)-V_t(x)-\frac{\tilde{l}}{t}\right)}  \\
0 &e^{n(g_+^t(x)-g_-^t(x))}
\end{pmatrix},\quad x\in\mathbb{R}.
\end{equation*}

\subsubsection{Opening of the lens}
We now perform a standard technique in the steepest decent method
\cite{BI}, \cite{DKV}, \cite{DKV2}. First note that, from
(\ref{eq:gineq}), we see that $J_T(x)$ becomes the following on the
interval $[\alpha_t,\beta_t]$.
\begin{equation*}
J_T(x)=\begin{pmatrix} e^{-n(g_+^t(x)-g_-^t(x))} & e^{2D_n(x)}  \\
0 &e^{n(g_+^t(x)-g_-^t(x))}
\end{pmatrix},\quad x\in [\alpha_t,\beta_t]
\end{equation*}
where $D_n(x)$ is the function
\begin{equation}\label{eq:Dn}
D_n(x)=\upsilon_t(x)\frac{n\delta t}{\log\delta t},\quad t>1,\quad
D_n(x)=0,\quad t\leq 1
\end{equation}
which is bounded on $[\alpha_t,\beta_t]$ under the double scaling
limit (\ref{eq:scale00}).

Then from (\ref{eq:gd}), the jump matrix $J_T(x)$ has the following
factorization on $[\alpha_t,\beta_t]$.
\begin{equation*}
J_T(x)=\begin{pmatrix} 1 & 0  \\
e^{n\left(V_t(x)-2g^t_-(x)+\frac{\tilde{l}}{t}\right)-2D_n(x)} &1
\end{pmatrix}\begin{pmatrix} 0 & e^{2D_n(x)}  \\
-e^{-2D_n(x)} &0
\end{pmatrix}\begin{pmatrix} 1 & 0  \\
e^{n\left(V_t(x)-2g^t_+(x)+\frac{\tilde{l}}{t}\right)-2D_n(x)} &1
\end{pmatrix}
\end{equation*}
\begin{figure}
\centering
\begin{overpic}[scale=.35, unit=1mm]{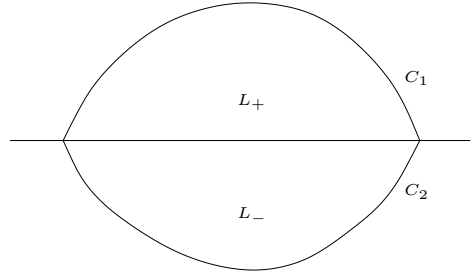}
\put(30,22){\tiny {$L_+$}}\put(30,7){\tiny {$L_-$}}\put(52,25){\tiny
{$C_1$}}\put(52,10){\tiny {$C_2$}}
\end{overpic}
\caption{The opening of the lens and different regions in the
lens.}\label{fig:lens}
\end{figure}
As in \cite{BI}, \cite{DKV}, \cite{DKV2}, we can open a lens around
the interval $[\alpha_t,\beta_t]$ as shown in Figure \ref{fig:lens}
and define the matrix $S(x)$ to be the following
\begin{equation}\label{eq:S}
\begin{split}
S(x)=\left\{
       \begin{array}{ll}
         T(x), & \hbox{$x$ outside the lens;} \\
         T(x)\begin{pmatrix} 1 & 0  \\
-e^{n\left(V_t(x)-2g^t(x)+\frac{\tilde{l}}{t}\right)-2D_n(x)} &1
\end{pmatrix}, & \hbox{$x\in L_+$;} \\
         T(x)\begin{pmatrix} 1 & 0  \\
e^{n\left(V_t(x)-2g^t(x)+\frac{\tilde{l}}{t}\right)-2D_n(x)} &1
\end{pmatrix}, & \hbox{$x\in L_-$.}
       \end{array}
     \right.
\end{split}
\end{equation}
Then the function $S(x)$ will satisfy the following Riemann-Hilbert
problem.
\begin{equation}\label{eq:RHS}
\begin{split}
&1. \quad \text{$S(x)$ is analytic in $\mathbb{C}/\mathbb{R}$};\\
&2. \quad S_+(x)=S_-(x)J_S(x),\quad x\in\mathbb{R};\\
&3. \quad S(x)=I+O(x^{-1}),\quad x\rightarrow\infty.
\end{split}
\end{equation}
where $J_S(x)$ is now defined by the following
\begin{equation}\label{eq:Sjump1}
\begin{split}
J_S(x)&=\begin{pmatrix} 1 & 0  \\
e^{n\left(V_t(x)-2g^t(x)+\frac{\tilde{l}}{t}\right)-2D_n(x)} &1
\end{pmatrix},\quad x\in C_1\cup C_2\\
J_S(x)&=\begin{pmatrix} 0 & e^{2D_n(x)}  \\
-e^{-2D_n(x)} &0\end{pmatrix},\quad x\in (\alpha_t,\beta_t)\\
J_S(x)&=\begin{pmatrix} 1 & e^{n\left(2g^t(x)-V_t(x)-\frac{\tilde{l}}{t}\right)+2D_n(x)}  \\
0 &1\end{pmatrix},\quad x\in (-\infty,\alpha_t)\\
J_S(x)&=\begin{pmatrix} e^{-2\pi iu^t} & e^{n\left(2g^t(x)-V_t(x)-l_t\right)+2D_n(x)}  \\
0 &e^{2\pi iu^t}\end{pmatrix},\quad x\in (\beta_t,x^{\ast})
\end{split}
\end{equation}
Then from (\ref{eq:gineq}), we see that for some large enough $n$
and $t$ close to 1 such that $\delta t\frac{n}{\log n}=o(1)$, we
have
\begin{equation*}
\begin{split}
&1.\quad \text{$D_n(x)$ is uniformly bounded on $[\alpha_t,\beta_t]$.}\\
&2.\quad
\text{$e^{n\left(2g^t(x)-V_t(x)-\frac{\tilde{l}}{t}\right)}\rightarrow
0$
on $\mathbb{R}/\left(\bigcup_{j=1}^3B_{\delta}^{r_j}\cup [\alpha_t,\beta_t]\right)$.}\\
&3. \quad
\text{$e^{n\left(V_t(x)-2g^t(x)+\frac{\tilde{l}}{t}\right)}\rightarrow
0$ on $\left(C_1\cup C_2\right)/\left(B_{\delta}^{-2}\cup
B_{\delta}^2\right)$.}
\end{split}
\end{equation*}
Therefore, in the double scaling limit, the jump matrix $J_S(x)$
behaves as
\begin{equation*}
\begin{split}
J_S(x)\rightarrow I,\quad \left(x\in \mathbb{R}\cup C_1\cup C_2\right)/\left(\bigcup_{j=1}^3B_{\delta}^{r_j}\cup[\alpha_t,\beta_t]\right),\\
J_S(x)=\begin{pmatrix} 0 & e^{2D_n(x)}  \\
-e^{-2D_n(x)} &0\end{pmatrix},\quad x\in(\alpha_t,\beta_t)
\end{split}
\end{equation*}
\subsection{Parametrix outside of the points $\alpha_t$, $\beta_t$ and
$x^{\ast}$}

We will now construct the parametrix outside of the singular points.
We would like to find a solution to the following Riemann-Hilbert
problem.
\begin{equation}\label{eq:sinf}
\begin{split}
&1. \quad \text{$S^{\infty}(x)$ is analytic in $\mathbb{C}/\left(\mathbb{R}\cup B_{\delta}^{x^{\ast}}\right)$};\\
&2. \quad S_+^{\infty}(x)=S_-^{\infty}(x)J^{\infty}(x),\quad x\in\mathbb{R};\\
&3. \quad S^{\infty}(x)=I+O(x^{-1}),\quad x\rightarrow\infty.
\end{split}
\end{equation}
where $J^{\infty}(x)$ is the following matrix-valued function.
\begin{equation}\label{eq:jinf}
\begin{split}
J^{\infty}(x)&=\begin{pmatrix} e^{-2\pi iu^t} & 0  \\
0 &e^{2\pi iu^t}\end{pmatrix},\quad x\in(\beta_t,x^{\ast})\\
J^{\infty}(x)&=\begin{pmatrix} 0 & e^{2D_n(x)}  \\
-e^{-2D_n(x)} &0\end{pmatrix},\quad x\in (\alpha_t,\beta_t).
\end{split}
\end{equation}
We should construct several scalar functions and use them to `dress'
the 1-cut parametrix constructed in \cite{DKV} so that it satisfies
the Riemann-Hilbert problem (\ref{eq:sinf}) with the jumps
(\ref{eq:jinf}).

\subsubsection{The limiting Abelian differential}\label{se:F}

The discussions in sections \ref{se:F} and \ref{se:D} are only
relevant when $t>1$. Let us first construct a function $F(x)$ with
the following jump discontinuities
\begin{equation*}
\begin{split}
F_+(x)&=-F_-(x),\quad x\in [\alpha_t,\beta_t]\\
F_+(x)&=F_-(x)+2\pi i, \quad x\in [\beta_t,x^{\ast}]\\
\end{split}
\end{equation*}
We have the following
\begin{lemma}\label{le:F}
The function $F(x)$ defined by
\begin{equation}\label{eq:F}
F(x)=\int_{\alpha_t}^{x}\frac{\sqrt{(x^{\ast}-\alpha_t)(x^{\ast}-\beta_t)}ds}{\sqrt{(s-\alpha_t)(s-\beta_t)}(s-x^{\ast})}=\log
\frac{\sqrt{\frac{x^{\ast}-\alpha_t}{x^{\ast}-\beta_t}}-\sqrt{\frac{x-\alpha_t}{x-\beta_t}}}
{\sqrt{\frac{x^{\ast}-\alpha_t}{x^{\ast}-\beta_t}}+\sqrt{\frac{x-\alpha_t}{x-\beta_t}}}
\end{equation}
satisfies the following scalar Riemann-Hilbert problem.
\begin{equation}\label{eq:RHF}
\begin{split}
&1. \quad \text{$F(x)$ is analytic in $\mathbb{C}/[\alpha_t,x^{\ast}]$};\\
&2. \quad F_+(x)=-F_-(x),\quad x\in [\alpha_t,\beta_t];\\
&3. \quad F_+(x)=F_-(x)+2\pi i,\quad x\in [\beta_t,x^{\ast}];\\
&4. \quad F(x)=\log (x-x^{\ast})+O(1),\quad x\rightarrow x^{\ast};\\
&5. \quad F(x)=F_0;
+O(x^{-1}),\quad x\rightarrow\infty;\\
&F_0=\log
\frac{\sqrt{\frac{x^{\ast}-\alpha_t}{x^{\ast}-\beta_t}}-1}{\sqrt{\frac{x^{\ast}-\alpha_t}{x^{\ast}-\beta_t}}+1},
\\
 &6. \quad
\text{$F(x)$ is bounded as $x$ approaches $\alpha_t$ or $\beta_t$.}
\end{split}
\end{equation}
where the square roots and logarithm are chosen such that the branch
cut is on the negative real axis. Also, to avoid ambiguity, we set
\begin{equation*}
\sqrt{\frac{x-\alpha_t}{x-\beta_t}}=\frac{\sqrt{(x-\alpha_t)(x-\beta_t)}}{(x-\beta_t)^2}
\end{equation*}
\end{lemma}
\begin{proof} First let us show that the only singularity of $F(x)$
is at $x^{\ast}$. The function $F(x)$ can become singular when the
argument becomes zero or infinity. The numerator and the denominator
of the argument can become zero if
\begin{equation*}
\begin{split}
\sqrt{\frac{x^{\ast}-\alpha_t}{x^{\ast}-\beta_t}}\pm\sqrt{\frac{x-\alpha_t}{x-\beta_t}}=0,
\end{split}
\end{equation*}
which implies
\begin{equation*}
\begin{split}
&\frac{x^{\ast}-\alpha_t}{x^{\ast}-\beta_t}=\frac{x-\alpha_t}{x-\beta_t},\\
&x=x^{\ast}.
\end{split}
\end{equation*}
Since the denominator is non-zero at $x=x^{\ast}$, we see that the
denominator does not vanish for all $x\in\mathbb{C}$. Near
$x=x^{\ast}$, we can expand the numerator and the denominator in a
power series of $x-x^{\ast}$.
\begin{equation*}
\frac{\sqrt{\frac{x^{\ast}-\alpha_t}{x^{\ast}-\beta_t}}-\sqrt{\frac{x-\alpha_t}{x-\beta_t}}}
{\sqrt{\frac{x^{\ast}-\alpha_t}{x^{\ast}-\beta_t}}+\sqrt{\frac{x-\alpha_t}{x-\beta_t}}}=c_0(x-x^{\ast})+O((x-x^{\ast})^2),
\end{equation*}
where $c_0$ is the following constant
\begin{equation*}
\begin{split}
c_0=-\frac{\frac{d}{dx}\left(\sqrt{\frac{x-\alpha_t}{x-\beta_t}}\right)|_{x=x^{\ast}}}{2\sqrt{\frac{x^{\ast}-\alpha_t}{x^{\ast}-\beta_t}}}
=\frac{\beta_t-\alpha_t}{4(x^{\ast}-\beta_t)(x^{\ast}-\alpha_t)}\neq
0.
\end{split}
\end{equation*}
Therefore near $x=x^{\ast}$, $F(x)$ behaves like
\begin{equation*}
F(x)=\log(x-x^{\ast})+O(1),\quad x\rightarrow x^{\ast},
\end{equation*}
this proves 4.

The other points where $F(x)$ can be singular are the points
$x=\alpha_t$, $\beta_t$ or $x=\infty$, where the argument may become
infinite. However, near $x=\alpha_t$, the function
$\sqrt{\frac{x-\alpha_t}{x-\beta_t}}$ remains finite and therefore
$F(x)$ does not have singularity near it. Let the argument inside
the logarithm of (\ref{eq:F}) be $\Phi(x)$.
\begin{equation}\label{eq:arg}
\begin{split}
\Phi(x)=\frac{\sqrt{\frac{x^{\ast}-\alpha_t}{x^{\ast}-\beta_t}}-\sqrt{\frac{x-\alpha_t}{x-\beta_t}}}
{\sqrt{\frac{x^{\ast}-\alpha_t}{x^{\ast}-\beta_t}}+\sqrt{\frac{x-\alpha_t}{x-\beta_t}}}.
\end{split}
\end{equation}
Then near $x=\beta_t$, $\Phi(x)$ behaves like
\begin{equation*}
\Phi(x)=\frac{\sqrt{\frac{(x^{\ast}-\alpha_t)(x-\beta_t)}{x^{\ast}-\beta_t}}-\sqrt{x-\alpha_t}}
{\sqrt{\frac{(x^{\ast}-\alpha_t)(x-\beta_t)}{x^{\ast}-\beta_t}}+\sqrt{x-\alpha_t}},
\end{equation*}
which is bounded and non-zero as $x\rightarrow \beta_t$.

We now consider the point $x=\infty$. Near $x=\infty$, we can
rewrite $\Phi(x)$ as
\begin{equation*}
\Phi(x)=\frac{\sqrt{\frac{(x^{\ast}-\alpha_t)}{x^{\ast}-\beta_t}}-\sqrt{\frac{1-\frac{\alpha_t}{x}}{1-\frac{\beta_t}{x}}}}
{\sqrt{\frac{(x^{\ast}-\alpha_t)}{x^{\ast}-\beta_t}}+\sqrt{\frac{1-\frac{\alpha_t}{x}}{1-\frac{\beta_t}{x}}}}.
\end{equation*}
Therefore the asymptotic behavior of $\Phi(x)$ near $x=\infty$ is
given by
\begin{equation*}
\Phi(x)=\frac{\sqrt{\frac{(x^{\ast}-\alpha_t)}{x^{\ast}-\beta_t}}-1}
{\sqrt{\frac{(x^{\ast}-\alpha_t)}{x^{\ast}-\beta_t}}+1}+O(x^{-1}).
\end{equation*}
Hence $F(x)$ has a singularity at $x=x^{\ast}$ only and this proves
5 and 6.

We will now study the jump discontinuities of $F(x)$. First note
that $F(x)$ can only have jump discontinuities outside
$[\alpha_t,\beta_t]$ if
\begin{equation*}
\Phi(x)=\frac{\sqrt{\frac{x^{\ast}-\alpha_t}{x^{\ast}-\beta_t}}-\sqrt{\frac{x-\alpha_t}{x-\beta_t}}}
{\sqrt{\frac{x^{\ast}-\alpha_t}{x^{\ast}-\beta_t}}+\sqrt{\frac{x-\alpha_t}{x-\beta_t}}}\in\mathbb{R}.
\end{equation*}
Simple calculations shows that $x\in\mathbb{R}$. Therefore $F(x)$
can only have jumps on the real axis.

We will first consider the jump on $[\alpha_t,\beta_t]$. On
$[\alpha_t,\beta_t]$ the square function
$\sqrt{\frac{x-\alpha_t}{x-\beta_t}}$ has the following jump
discontinuity.
\begin{equation*}
\left(\sqrt{\frac{x-\alpha_t}{x-\beta_t}}\right)_+=-\left(\sqrt{\frac{x-\alpha_t}{x-\beta_t}}\right)_-.
\end{equation*}
Hence $F_-(x)$ is given by
\begin{equation*}
\begin{split}
F_-(x)=\log
\frac{\sqrt{\frac{x^{\ast}-\alpha_t}{x^{\ast}-\beta_t}}-\left(\sqrt{\frac{x-\alpha_t}{x-\beta_t}}\right)_-}
{\sqrt{\frac{x^{\ast}-\alpha_t}{x^{\ast}-\beta_t}}+\left(\sqrt{\frac{x-\alpha_t}{x-\beta_t}}\right)_-}=\log
\frac{\sqrt{\frac{x^{\ast}-\alpha_t}{x^{\ast}-\beta_t}}+\left(\sqrt{\frac{x-\alpha_t}{x-\beta_t}}\right)_+}
{\sqrt{\frac{x^{\ast}-\alpha_t}{x^{\ast}-\beta_t}}-\left(\sqrt{\frac{x-\alpha_t}{x-\beta_t}}\right)_+}=-F_+(x).
\end{split}
\end{equation*}
This proves property 2.

We now consider the jump on $[\beta_t,x^{\ast}]$. The function
$\Phi(x)$ in (\ref{eq:arg}) is real on
$\mathbb{R}/[\alpha_t,\beta_t]$ and we need to show that it is
negative on $(\beta_t,x^{\ast})$ and positive elsewhere. Let us
first consider a point $x\in (\beta_t,\infty)$. The the denominator
of $\Phi(x)$ is positive for $x\in (\beta_t,\infty)$. To study the
signs of the numerator, let us consider its derivative
\begin{equation}\label{eq:der}
\begin{split}
\frac{d}{dx}\left(\sqrt{\frac{x^{\ast}-\alpha_t}{x^{\ast}-\beta_t}}-\sqrt{\frac{x-\alpha_t}{x-\beta_t}}\right)
=\frac{1}{2}\sqrt{\frac{x-\beta_t}{x-\alpha_t}}\frac{\beta_t-\alpha_t}{(x-\beta_t)^2}>0,\quad
x\in (\beta_t,\infty),
\end{split}
\end{equation}
we see that the numerator is a strictly increasing function on
$(\beta_t,\infty)$. Since it vanishes at $x=x^{\ast}$, we see that
it is negative on $(\beta_t,x^{\ast})$ and positive on
$(x^{\ast},\infty)$.

Now let us look at the sign of $\Phi(x)$ on $(-\infty,\alpha_t)$.
First note that, on $(\beta_t,\infty)$, the square root
$\sqrt{\frac{x-\alpha_t}{x-\beta_t}}$ is positive and strictly
decreasing and hence it is greater than 1 on $(\beta_t,\infty)$. In
particular, we have
\begin{equation*}
\sqrt{\frac{x^{\ast}-\alpha_t}{x^{\ast}-\beta_t}}>1.
\end{equation*}
Now let $x\in(-\infty,\alpha_t)$. In this region, the square root
$\sqrt{\frac{x-\alpha_t}{x-\beta_t}}$ is positive and strictly
decreasing from (\ref{eq:der}). Near $-\infty$, it approaches 1
while at $\alpha_t$, it becomes zero. Therefore on
$(-\infty,\alpha_t)$, it takes values between 0 and 1.

Therefore we have
\begin{equation*}
\begin{split}
&\sqrt{\frac{x^{\ast}-\alpha_t}{x^{\ast}-\beta_t}}-\sqrt{\frac{x-\alpha_t}{x-\beta_t}}>1-1=0,\quad
x\in
(-\infty,\alpha_t)\\
&\sqrt{\frac{x^{\ast}-\alpha_t}{x^{\ast}-\beta_t}}+\sqrt{\frac{x-\alpha_t}{x-\beta_t}}>0,\quad
x\in (-\infty,\alpha_t).
\end{split}
\end{equation*}
Hence the $\Phi(x)$ is positive on $(-\infty,\alpha_t)$.
Summarizing, we have
\begin{equation*}
\begin{split}
&\Phi(x)>0,\quad
x\in (-\infty,\alpha_t)\cup(x^{\ast},\infty)\\
&\Phi(x)<0,\quad x\in (\beta_t,x^{\ast}).
\end{split}
\end{equation*}
This proves 3.

Since $F(x)$ cannot have any jump discontinuities and singularities
other than the ones that are considered here, property 1. is true.
\end{proof}
\begin{remark}The function $F(x)$ can be thought of as the limit of
an Abelian integral on an elliptic curve. Let us consider the
following elliptic curve
\begin{equation}\label{eq:ell}
z^2=(x-\alpha_t)(x-\beta_t)((x-x^{\ast})^2-\sigma_t^2)
\end{equation}
\begin{figure}
\centering
\begin{overpic}[scale=1,unit=1mm]{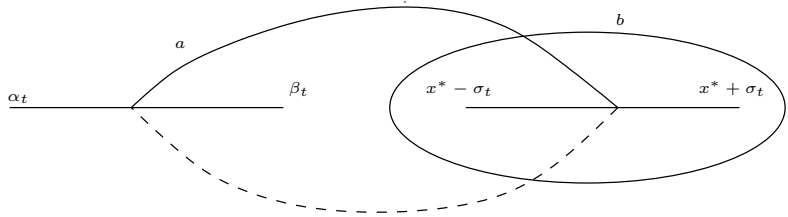}
\put(0,15){\tiny {$\alpha_t$}}\put(37,16){\tiny
{$\beta_t$}}\put(55,16){\tiny
{$x^{\ast}-\sigma_t$}}\put(90,16){\tiny{
$x^{\ast}+\sigma_t$}}\put(80,25){\tiny{$ b$}}\put(22,22){\tiny
{$a$}}
\end{overpic}
\caption{The $a$ and $b$ cycle of the elliptic curve
(\ref{eq:ell}).}\label{fig:curve}
\end{figure}
and define the $a$ and $b$-cycles of this curve as in Figure
\ref{fig:curve}. Then the normalized holomorphic Abelian
differential on this curve is given by
\begin{equation}\label{eq:normform}
\Omega(x)=\frac{Cdx}{\sqrt{(x-\alpha_t)(x-\beta_t)((x-x^{\ast})^2-\sigma_t^2)}}
\end{equation}
for some constant $C$ such that
\begin{equation*}
\oint_{a}\Omega(x)=1
\end{equation*}
In the limit $\sigma_t\rightarrow 0$, the curve becomes degenerate
and the Abelian integral $\int^x\Omega(s)$ degenerates into the
function $F(x)$.
\end{remark}

\subsubsection{Scalar function with jump $D_n(x)$}\label{se:D}
We will now seek a scalar function $K_n(x)$ that is bounded at
infinity and has jump $2D_n(x)$ on $[\alpha_t,\beta_t]$. (cf. the
Szego function used in \cite{KV}, \cite{V})

We shall construct a function $K(x)$ that satisfies the following
Riemann-Hilbert problem.
\begin{equation}\label{eq:RHK}
\begin{split}
&1. \quad \text{$K(x)$ is analytic on $\mathbb{C}/[\alpha_t,\beta_t]$};\\
&2. \quad K_+(x)=-K_-(x)+2D_n(x),\quad x\in [\alpha_t,\beta_t];\\
&3. \quad K(x)=K_0+ O(x^{-1}),\quad x\rightarrow\infty,\\
&K_0=-\frac{1}{2\pi
i}\int_{\alpha_t}^{\beta_t}\frac{2D_n(s)ds}{\left(\sqrt{(s-\alpha_t)(s-\beta_t)}\right)_+}.
\end{split}
\end{equation}
This function can be constructed by the use of Cauchy transform
easily \cite{Mus}.
\begin{equation}\label{eq:K}
\begin{split}
K(x)&=\frac{\sqrt{(x-\alpha_t)(x-\beta_t)}}{2\pi
i}\int_{\alpha_t}^{\beta_t}\frac{2D_n(s)ds}{\left(\sqrt{(s-\alpha_t)(s-\beta_t)}\right)_+(s-x)}
\end{split}
\end{equation}
\begin{lemma}\label{le:K}
The function $K(x)$ defined in (\ref{eq:K}) satisfies the
Riemann-Hilbert problem (\ref{eq:RHK}).
\end{lemma}
\begin{proof}
Let $C(f)$ be the Cauchy transform
\begin{equation*}
C(f)(x)=\frac{1}{2\pi i}\int_{\alpha_t}^{\beta_t}\frac{f(s)ds}{s-x}
\end{equation*}
Then from the Plemelj formula (See, e.g.\cite{Mus}), we have
\begin{equation*}
\begin{split}
C_{\pm}(f)(x)=\pm\frac{1}{2}f+\frac{1}{2\pi
i}\int_{\alpha_t}^{\beta_t}\frac{f(s)ds}{s-x},\quad x\in
[\alpha_t,\beta_t],
\end{split}
\end{equation*}
where the principal value is taken in the integral on the right hand
side. Taking into account the change of sign of
$\sqrt{(x-\alpha_t)(x-\beta_t)}$ across $[\alpha_t,\beta_t]$, we
have
\begin{equation*}
\begin{split}
K_+(x)+K_-(x)&=\left(\sqrt{(x-\alpha_t)(x-\beta_t)}\right)_+\frac{2D_n(x)}{\left(\sqrt{(x-\alpha_t)(x-\beta_t)}\right)_+}\\
&=2D_n(x).
\end{split}
\end{equation*}
To see that $K(x)$ has the desired property at $x=\infty$, let us
write the factor $\frac{1}{s-x}$ in a power series near $x=\infty$.
\begin{equation*}
\frac{1}{s-x}=-\frac{1}{x}\left(1+\frac{s}{x}+O(x^{-2})\right).
\end{equation*}
Therefore the function $K(x)$ behaves as follows as
$x\rightarrow\infty$.
\begin{equation*}
\begin{split}
K(x)&=-\frac{\sqrt{(x-\alpha_t)(x-\beta_t)}}{2\pi
ix}\left(\int_{\alpha_t}^{\beta_t}\frac{2D_n(s)ds}{\left(\sqrt{(s-\alpha_t)(s-\beta_t)}\right)_+}+O(x^{-1})\right)\\
&=-\frac{1}{2\pi
i}\int_{\alpha_t}^{\beta_t}\frac{2D_n(s)ds}{\left(\sqrt{(s-\alpha_t)(s-\beta_t)}\right)_+}+O(x^{-1})
\end{split}
\end{equation*}
which is property 3. in (\ref{eq:RHK}).
\end{proof}
\subsubsection{Parametrix outside of special points}
We are now in a position to construct the parametrix outside of the
special points. First let us consider the following matrix.
\begin{equation}\label{eq:P}
\Pi(x)=\begin{pmatrix} \frac{\gamma(x)+\gamma(x)^{-1}}{2} & \frac{\gamma(x)-\gamma(x)^{-1}}{2i}  \\
-\frac{\gamma(x)-\gamma(x)^{-1}}{2i}
&\frac{\gamma(x)+\gamma(x)^{-1}}{2}\end{pmatrix}
\end{equation}
where $\gamma(x)=\gamma_t(x)$ is defined by
\begin{equation*}
\gamma(x)=\left(\frac{x-\beta_t}{x-\alpha_t}\right)^{\frac{1}{4}}
\end{equation*}
Recall that this matrix satisfies the following Riemann-Hilbert
problem \cite{DKV}, \cite{D}, \cite{BI}.
\begin{equation*}
\begin{split}
&1. \quad \text{$\Pi(x)$ is analytic on $\mathbb{C}/[\alpha_t,\beta_t]$};\\
&2. \quad \Pi_+(x)=\Pi_-(x)\begin{pmatrix} 0 & 1  \\
-1&0\end{pmatrix},\quad x\in [\alpha_t,\beta_t];\\
&3. \quad \Pi(x)=I+ O(x^{-1}),\quad x\rightarrow\infty.
\end{split}
\end{equation*}
We can now combine $\Pi(x)$, $K(x)$ and $F(x)$ to form our
parametrix outside the special points. The main result is the
following.
\begin{proposition}\label{pro:global}
The matrix $S^{\infty}(x)$ defined by
\begin{equation}\label{eq:sinf1}
S^{\infty}(x)=e^{\left(K_0+(u^t-\overline{u}^t)F_0\right)\sigma_3}\Pi(x)e^{-\left(K(x)+(u^t-\overline{u}^t)F(x)\right)\sigma_3}
\end{equation}
satisfies the Riemann-Hilbert problem (\ref{eq:sinf}) and
(\ref{eq:jinf}). In particular, when $t\leq 1$, both
$u^t-\overline{u}^t$ and $K(x)$ vanishes and the above equation
reduces to
\begin{equation}\label{eq:sin2}
S^{\infty}(x)=\Pi(x),\quad t\leq 1.
\end{equation}
\end{proposition}
\begin{proof} First let us consider the asymptotic behavior near
$x=\infty$. We have
\begin{equation*}
\begin{split}
S^{\infty}(x)&=e^{\left(K_0+(u^t-\overline{u}^t)F_0\right)\sigma_3}\left(I+O(x^{-1})\right)e^{-\left(K_0+(u^t-\overline{u}^t)F_0\right)\sigma_3}\\
&=\left(I+O(x^{-1})\right)
\end{split}
\end{equation*}
This proves 3. in (\ref{eq:sinf}). Next we consider the jump
discontinuities. It is given by
\begin{equation}\label{eq:ju}
\begin{split}
S^{\infty}_+(x)&=S^{\infty}_-(x)\begin{pmatrix}e^{A_1(x)} & 0  \\
0&e^{-A_1(x)}\end{pmatrix},\quad x\in
(-\infty,x^{\ast})/[\alpha_t,\beta_t]\\
S^{\infty}_+(x)&=S^{\infty}_-(x)\begin{pmatrix}0 & e^{A_2(x)}  \\
-e^{-A_2(x)}&0\end{pmatrix},\quad x\in (\alpha_t,\beta_t)
\end{split}
\end{equation}
where $A_i(x)$ is given by
\begin{equation*}
\begin{split}
A_1(x)&=(u^t-\overline{u}^t)\left(-F_+(x)+F_-(x)\right),\quad x\in (\beta_t,x^{\ast}),\\
A_2(x)&=(\overline{u}^t-u^t)(F_+(x)+F_-(x))+K_+(x)+K_-(x),\quad x\in
(\alpha_t,\beta_t).
\end{split}
\end{equation*}
From (\ref{eq:RHF}) and (\ref{eq:RHK}), we see that $A_1(x)$ and
$A_2(x)$ are in fact the following
\begin{equation*}
\begin{split}
A_1(x)&=-2\pi i(u^t-\overline{u}^t),\quad x\in (\beta_t,x^{\ast})\\
A_2(x)&=2D_n(x),\quad x\in (\alpha_t,\beta_t).
\end{split}
\end{equation*}
Since $\overline{u}^t$ is an integer, we see that
\begin{equation*}
\begin{split}
e^{A_1(x)}&=e^{-2\pi iu^t},\quad x\in (\beta_t,x^{\ast}).
\end{split}
\end{equation*}
Substituting these back into (\ref{eq:ju}), we see that the matrix
$S^{\infty}(x)$ does indeed satisfy the jump conditions
(\ref{eq:jinf}).
\end{proof}

\begin{remark}The appearance of the degenerate Abelian integral
$F(x)$ (\ref{eq:F}) in the parametrix (\ref{eq:sinf1}) comes from
the appearance of the elliptic theta function in the 1-cut
parametrix. Let the 1-cut parametrix outside of the special points
be $M^{\infty}(x)$. Then $M^{\infty}(x)$ is given by \cite{DKV}:
\begin{equation*}
\begin{split}
M^{\infty}(x)&=H\begin{pmatrix}\frac{\gamma+\gamma^{-1}}{2}\frac{\theta(W(x)-(u^t-\overline{u}^t)+d)}{\theta(W(x)+d)}
&\frac{\gamma-\gamma^{-1}}{-2i}\frac{\theta(-W(x)-(u^t-\overline{u}^t)+d)}{\theta(-W(x)+d)}\\
\frac{\gamma-\gamma^{-1}}{2i}\frac{\theta(W(x)-(u^t-\overline{u}^t)-d)}{\theta(W(x)-d)}&
\frac{\gamma+\gamma^{-1}}{2}\frac{\theta(-W(x)-(u^t-\overline{u}^t)-d)}{\theta(-W(x)+d)}\end{pmatrix},\quad
\Im(z)>0\\
M^{\infty}(x)&=H\begin{pmatrix}\frac{\gamma-\gamma^{-1}}{2i}\frac{\theta(-W(x)-(u^t-\overline{u}^t)+d)}{\theta(-W(x)+d)}
&-\frac{\gamma+\gamma^{-1}}{2}\frac{\theta(W(x)-(u^t-\overline{u}^t)+d)}{\theta(W(x)+d)}\\
\frac{\gamma+\gamma^{-1}}{2}\frac{\theta(-W(x)-(u^t-\overline{u}^t)-d)}{\theta(W(x)+d)}&
-\frac{\gamma-\gamma^{-1}}{2i}\frac{\theta(W(x)-(u^t-\overline{u}^t)-d)}{\theta(W(x)-d)}\end{pmatrix},\quad
\Im(z)<0.
\end{split}
\end{equation*}
for some scalar constant $d$ and constant diagonal matrix $H$, where
$W(x)$ is the Abelian integral $\int^x\Omega(s)$ (\ref{eq:normform})
and $\theta(z)$ is the elliptic theta function. One can think of the
parametrix $S^{\infty}(x)$ as a degenerate version of
$M^{\infty}(x)$ as the branch cut
$[x^{\ast}-\sigma_t,x^{\ast}+\sigma_t]$ in (\ref{eq:ell}) is closing
up and the curve degenerates into a genus zero curve. In this case,
the Abelian integral degenerates into the function $F(x)$ and while
the theta function degenerates into an exponential function. In the
multi-cut case, one could apply the analysis similar to those in
\cite{BBI}, \cite{IMM} to obtain degenerate hyper-elliptic theta
functions and use them to construct the suitable parametrix.
\end{remark}

\subsection{Parametrix near the edge points $\alpha_t$ and $\beta_t$}

At the edge points $\alpha_t$ and $\beta_t$ the approximated density
$\tilde{\rho}^t(x)$ vanishes like a square root and the local
parametrices $S^{\pm 2}(x)$ near these points can be constructed by
the use of Airy-functions. Such construction has been done many
times in the literature and we should not repeat the details here.
An interested reader can consult \cite{DKV}, \cite{DKV2}, \cite{BI}
for example.
\subsection{Local parametrix near the critical point $x^{\ast}$ for
$t>1$}\label{se:lpara} We will now consider the parametrix near the
critical point $x^{\ast}$. As in \cite{Ey}, the parametrix will be
constructed out of the monic orthogonal polynomial
$\pi_{\overline{u}^t}^{\nu}(\zeta)$ of degree $\overline{u}^t$ and
weight $e^{-\zeta^{2\nu}}$, where $2\nu$ is the order of vanishing
of $2h(x)-V(x)+l$ at $x^{\ast}$.

We would like to construct a parametrix $S^{x^{\ast}}(x)$ in
$B_{\delta}^{x^{\ast}}$ such that
\begin{equation}\label{eq:lparast}
\begin{split}
&1.\quad \text{$S^{x^{\ast}}(x)$ is analytic in
$B_{\delta}^{x^{\ast}}/\left(B_{\delta}^{x^{\ast}}\cap\mathbb{R}\right)$};\\
&2. \quad S^{x^{\ast}}_+(x)=S^{x^{\ast}}_-(x)J_S(x),\quad x\in
B_{\delta}^{x^{\ast}}\cap\mathbb{R}; \\
&3. \quad \text{$S^{x^{\ast}}(x)=\left(I+o(1)\right)S^{\infty}(x)$
as $n\rightarrow\infty$, $t\rightarrow 1$, uniformly in $\p
B_{\delta}^{x^{\ast}}$}.
\end{split}
\end{equation}

\subsubsection{Conformal map in
$B_{\delta}^{x^{\ast}}$}\label{se:cont+}

Let us define a conformal map $\zeta=f(x)$ that maps the
neighborhood $B_{\delta}^{x^{\ast}}$ into the complex plane, such
that, as $n\rightarrow\infty$, the boundary of
$B_{\delta}^{x^{\ast}}$ is mapped into infinity. We will define
$\zeta$ as follows.
\begin{equation}\label{eq:zeta}
\zeta=f(x)=\left(-n\left(2h(x)-V(x)-l\right)\right)^{\frac{1}{2\nu}},
\end{equation}
where the $\frac{1}{2\nu}$th-root is chosen such that the intervals
$[x^{\ast}-\delta,x^{\ast}]$ and $[x^{\ast},x^{\ast}+\delta]$ are
mapped onto the negative and positive real axis respectively. This
is possible because $h(x)-\frac{V(x)}{2}-\frac{l}{2}$ vanishes to
order $2\nu$ at $x^{\ast}$ and that it is real and negative on the
interval $[x^{\ast}-\delta,x^{\ast}+\delta]$ due to
(\ref{eq:thineq}).

Since $h(x)-\frac{V(x)}{2}-\frac{l}{2}$ vanishes to order $2\nu$ at
$x^{\ast}$, the function $\zeta$ is of the form
\begin{equation}\label{eq:varphi}
\zeta=n^{\frac{1}{2\nu}}(x-x^{\ast})\varphi(x)
\end{equation}
such that $\varphi(x)$ is independent on $n$ and
$\varphi(x^{\ast})\neq 0$. By choosing $\delta t$ and $\delta$
smaller if necessary, we can assume that $\varphi(x)$ and hence
$\zeta$ is conformal inside the neighborhood
$B_{\delta}^{x^{\ast}}$.

Then $\zeta$ maps the neighborhood $B_{\delta}^{x^{\ast}}$ into the
complex $\zeta$-plane such that the boundary of
$B_{\delta}^{x^{\ast}}$ is mapped into infinity.

Let us now define the constant $Z_t$ and function $\tau_t(x)$ by
\begin{equation}\label{eq:tauc}
\begin{split}
Z_t&=n\left(g_1^t(x^{\ast})-\frac{V(x^{\ast})}{2t}-\frac{\tilde{l}}{2t}\right)
-u^t\left(\log\varphi(x^{\ast})+\frac{1}{2\nu}\log n\right),\\
\tau_t(x)&=\frac{n\left(2g_1^t(x)-\frac{V(x)}{t}-\frac{\tilde{l}}{t}\right)-2Z_t+\zeta^{2\nu}-2u^t\log
\left(n^{\frac{1}{2\nu}}\varphi(x)\right)}{\zeta},\\
g_1^t(x)&=\int_{\alpha_t}^{\beta_t}\left(\tilde{\rho}^t(s)-\iota_t\frac{\delta
t}{\log\delta
t}\frac{8\pi\sqrt{(s-\alpha_t)(\beta_t-s)}}{(\alpha_t+\beta_t)^2}\right)\log(x-s)ds.
\end{split}
\end{equation}
Note that $\tau_t(x)$ does not have a pole at $x=x^{\ast}$ and that
by taking $\delta t$ and $\delta$ smaller if necessary, we can
assume that $\tau_t(x)$ is analytic inside $B_{\delta}^{x^{\ast}}$.

Then it is easy to see that $\zeta$, $\tau_t(x)$ and $Z_t$ satisfy
\begin{equation}\label{eq:jumpz}
n\left(g^t(x)-\frac{V(x)}{2t}-\frac{\tilde{l}}{2t}\right)=-\frac{\zeta^{2\nu}}{2}+\frac{\tau_t(x)\zeta}{2}+u^t\log\zeta+Z_t.
\end{equation}
Moreover, we have the following
\begin{proposition}\label{pro:etau}
As $n\rightarrow\infty$ under the scaling (\ref{eq:scale00}), the
constant $Z_t$ and function $\tau_t(x)$ are of order
\begin{equation}\label{eq:etauord}
\begin{split}
Z_t&=-\frac{u^t}{2\nu}\log\log n+O(1),\\
\tau_t(x)&=O\left(\frac{\log n}{n^{\frac{1}{2\nu}}}\right),
\end{split}
\end{equation}
uniformly in $B_{\delta}^{x^{\ast}}$.
\end{proposition}
\begin{proof} The first part of the proposition follows from
(\ref{eq:ge}). By (\ref{eq:ge}) and (\ref{eq:tauc}), we have
\begin{equation*}
Z_t=n\delta
t\phi(x^{\ast})-u^t\left(\log\varphi(x^{\ast})+\frac{1}{2\nu}\log
n\right)+O\left(n\frac{\delta t}{\log\delta t}\right).
\end{equation*}
Then by using (\ref{eq:ut}) to eliminate $n\delta t$, we obtain the
following,
\begin{equation}\label{eq:eord}
Z_t=-\frac{u^t}{2\nu}\left(\log n+\log\delta
t\right)-u^t\log\varphi(x^{\ast})+O\left(n\frac{\delta t}{\log\delta
t}\right).
\end{equation}
Now by taking the logarithm of (\ref{eq:scale00}), we see that,
\begin{equation*}
\log n+\log\delta t=\log\log n+O(1).
\end{equation*}
Hence we obtain
\begin{equation*}
Z_t=-\frac{u^t}{2\nu}\log\log n+O(1).
\end{equation*}
This proves the 1st equation in (\ref{eq:etauord}).

Given the first equation in (\ref{eq:etauord}), the second equation
in (\ref{eq:etauord}) now follows easily from the definition of
$\zeta$ (\ref{eq:zeta}), $\tau_t(x)$ (\ref{eq:tauc}) and the
relation (\ref{eq:ge0}) from corollary \ref{cor:ge}.
\end{proof}
\subsubsection{Construction of the parametrix}

We should now construct the parametrix by using orthogonal
polynomials with weight
$\exp\left(-\zeta^{2\nu}+\tau_t(x)\zeta\right)$.

Let $\pi_{k}^{\nu}(\zeta,\tau)$ be the monic orthogonal polynomial
of degree $k$ with respect to the weight $-\zeta^{2\nu}+\tau\zeta$,
\begin{equation*}
\int_{\mathbb{R}}\pi_k^{\nu}(\zeta)\pi_j^{\nu}(\zeta)\exp\left(-\zeta^{2\nu}+\tau\zeta\right)d\zeta=h_k^{\nu}(\tau)\delta_{kj},
\end{equation*}
where $h_k^{\nu}(\tau)$ is the normalization constant as a function
of $\tau$. Let us denote by $\Psi^{\nu}(\zeta,s)$ the following
matrix constructed from the orthogonal polynomial
$\pi_{\overline{u}^t}^{\nu}(\zeta,\tau)$.
\begin{equation}\label{eq:psi}
\Psi^{\nu}(\zeta,\tau)=\begin{pmatrix}
\pi_{\overline{u}^t}^{\nu}(\zeta,\tau) & \frac{1}{2\pi
i}\int_{\mathbb{R}}\frac{\pi_{\overline{u}^t}^{\nu}(s,\tau)
\exp\left(-\zeta^{2\nu}+\tau\zeta\right)}{s-x}ds  \\
\kappa_{\overline{u}^t-1}^{\nu}(\tau)\pi_{\overline{u}^t-1}^{\nu}(\zeta,\tau)
& \frac{\kappa_{\overline{u}^t-1}^{\nu}(\tau)}{2\pi
i}\int_{\mathbb{R}}\frac{\pi_{\overline{u}^t-1}^{\nu}(s,\tau)\exp\left(-\zeta^{2\nu}+\tau\zeta\right)}{s-x}ds
\end{pmatrix},
\end{equation}
where $\kappa_{\overline{u}^t-1}^{\nu}(\tau)=-\frac{2\pi
i}{h_{\overline{u}^t-1}^{\nu}(\tau)}$.

Then the matrix $\Psi^{\nu}(\zeta,\tau)$ satisfies the following
Riemann-Hilbert problem.
\begin{equation}\label{eq:psiRHP}
\begin{split}
&1. \quad \text{$\Psi^{\nu}(\zeta,\tau)$ is analytic on $\mathbb{C}/\mathbb{R}$};\\
&2.\quad \Psi^{\nu}_+(\zeta,\tau)=\Psi^{\nu}_-(\zeta,\tau)\begin{pmatrix} 1 & \exp\left(-\zeta^{2\nu}+\tau\zeta\right)  \\
0 & 1
\end{pmatrix},\quad \zeta\in\mathbb{R};\\
&3. \quad
\Psi^{\nu}(\zeta,\tau)=\left(I+O(\zeta^{-1})\right)\begin{pmatrix}
\zeta^{\overline{u}^t} & 0
\\ 0 & \zeta^{-\overline{u}^t} \end{pmatrix},\quad
\zeta\rightarrow\infty.
\end{split}
\end{equation}
Let $E(x)$ be the following matrix-valued function,
\begin{equation}\label{eq:Ex}
E(x)=S^{\infty}(x)\zeta^{(u^t-\overline{u}^t)\sigma_3}e^{Z_t\sigma_3}
\end{equation}
Then from (\ref{eq:sinf1}) we see that
\begin{equation}\label{eq:Ean}
E(x)=e^{\left(K_0+(u^t-\overline{u}^t)F_0\right)\sigma_3}\Pi(x)e^{-K(x)\sigma_3}\left(e^{-(u^t-\overline{u}^t)F(x)\sigma_3}
\zeta^{(u^t-\overline{u}^t)\sigma_3}\right)e^{Z_t\sigma_3}.
\end{equation}
From property 4. of (\ref{eq:RHF}), we see that the factor
\begin{equation*}
e^{-(u^t-\overline{u}^t)F(x)\sigma_3}
\zeta^{(u^t-\overline{u}^t)\sigma_3}
\end{equation*}
is analytic inside $B_{\delta}^{x^{\ast}}$. Then, since both $K(x)$
and $\Pi(x)$ are analytic inside $B_{\delta}^{x^{\ast}}$, the
function $E(x)$ is analytic inside $B_{\delta}^{x^{\ast}}$. Hence we
have
\begin{proposition}\label{pro:lpara}
Let the matrix $S^{x^{\ast}}(x)$ be
\begin{equation}\label{eq:sxast}
S^{x^{\ast}}(x)=E(x)\Psi^{\nu}(\zeta,\tau_t(x))\zeta^{-u^t\sigma_3}e^{-Z_t\sigma_3},\quad
x\in B_{\delta}^{x^{\ast}},
\end{equation}
Then, under the double scaling limit (\ref{eq:scale00}),
$S^{x^{\ast}}(x)$ satisfies the conditions
\begin{equation}\label{eq:lRH}
\begin{split}
&1.\quad \text{$S^{x^{\ast}}(x)$ is analytic in
$B_{\delta}^{x^{\ast}}/\left(B_{\delta}^{x^{\ast}}\cap\mathbb{R}\right)$};\\
&2. \quad S^{x^{\ast}}_+(x)=S^{x^{\ast}}_-(x)J_S(x),\quad x\in
B_{\delta}^{x^{\ast}}\cap\mathbb{R}; \\
&3. \quad \text{$S^{x^{\ast}}(x)=\left(I+O\left(\frac{(\log
n)^{-\frac{u^t}{2\nu}}}{n^{\frac{1-2|u^t-\overline{u}^t|}{2\nu}}}\right)\right)S^{\infty}(x)$
as $n\rightarrow\infty$, uniformly in $\p B_{\delta}^{x^{\ast}}$}.
\end{split}
\end{equation}
\end{proposition}
\begin{proof} The properties 1. and 2. follows immediately from
(\ref{eq:jumpz}) and property 2. of (\ref{eq:psiRHP}). We should now
prove property 3.

At the boundary of $B_{\delta}^{x^{\ast}}$, we have
$\zeta\rightarrow\infty$ and hence the function $S^{x^{\ast}}(x)$
behaves as
\begin{equation*}
S^{x^{\ast}}(x)=S^{\infty}(x)\zeta^{(u^t-\overline{u}^t)\sigma_3}e^{Z_t\sigma_3}\left(I+O(\zeta^{-1})\right)e^{-Z_t\sigma_3}\zeta^{(\overline{u}^t-u^t)\sigma_3}
,\quad \zeta\rightarrow\infty.
\end{equation*}
From (\ref{eq:zeta}), we see that
$\zeta^{-1}=O(n^{-\frac{1}{2\nu}})$ at the boundary of
$B_{\delta}^{x^{\ast}}$, hence the above equation becomes
\begin{equation*}
\begin{split}
S^{x^{\ast}}(x)&=S^{\infty}(x)\left(I+O\left(\frac{(\log
n)^{\frac{u^t}{2\nu}}}{n^{\frac{1-2|u^t-\overline{u}^t|}{2\nu}}}\right)\right),\\
&=\left(I+O\left(\frac{(\log
n)^{\frac{u^t}{2\nu}}}{n^{\frac{1-2|u^t-\overline{u}^t|}{2\nu}}}\right)\right)S^{\infty}(x).
\end{split}
\end{equation*}
where the second equality follows from the fact that $S^{\infty}(x)$
is bounded in $B_{\delta}^{x^{\ast}}$ as $n\rightarrow\infty$. This
proves the proposition.
\end{proof}

\subsection{Local parametrix near the critical point $x^{\ast}$ for $t\leq1$}

We will now construct the parametrix in $B_{\delta}^{x^{\ast}}$ when
$t\leq1$. In this case, the parametrix can be constructed from the
Cauchy transform \cite{DKV2}.

\subsubsection{Conformal map in $B_{\delta}^{x^{\ast}}$}

We will use the same conformal map (\ref{eq:zeta}) defined in
section \ref{se:cont+}. However, the function $\tau_t(x)$ and the
constant $Z_t$ are now defined to be
\begin{equation}\label{eq:Ztau1}
\begin{split}
Z_t&=n\left(g^t(x^{\ast})-\frac{V_t(x^{\ast})}{2}-\frac{l_t}{2}\right)\\
\tau_t(x)&=\frac{n\left(2g^t(x)-V_t(x)-l_t\right)-2Z_t+\zeta^{2\nu}}{\zeta}.
\end{split}
\end{equation}
Then by (\ref{eq:scalep}) and the Buyarov-Rakhmanov equation
(\ref{eq:br1}), we see that $Z_t$ is of order
\begin{equation}\label{eq:Zt}
Z_t=O(n\delta t)=O(n^{1-k})
\end{equation}
As we see in (\ref{eq:critker1}), the limiting kernel will be of
order $e^{Z_t}$. However, one should bear in mind that $Z_t$ is
negative and therefore the term $e^{Z_t}$ is bounded as
$n\rightarrow\infty$.

We can now deduce the order of $\tau_t(x)$ from (\ref{eq:br1})
,(\ref{eq:scalep}) and (\ref{eq:Zt}).
\begin{equation}\label{eq:ordtau}
\tau_t(x)=2n^{1-k-\frac{1}{2\nu}}\frac{U_-\left(\phi(x)-\phi(x^{\ast})\right)}{(x-x^{\ast})\varphi(x)}+o(1)
\end{equation}
From the definition (\ref{eq:Ztau1}), we see that $\zeta$, $Z_t$ and
$\tau_t(x)$ together satisfies the following
\begin{equation}\label{eq:jumpz-}
n\left(g^t(x)-\frac{V_t(x)}{2}-\frac{l_t}{2}\right)=-\frac{\zeta^{2\nu}}{2}+\frac{\tau_t(x)\zeta}{2}+Z_t
\end{equation}

\subsubsection{Construction of the parametrix}
Let us now construct the local parametrix by using Cauchy transform
(cf. \cite{DKV2}). Let $\Psi(\zeta,\tau_t(x))$ be the following
matrix.
\begin{equation}\label{eq:psi-}
\Psi(\zeta,\tau_t(x))=\begin{pmatrix} 1 & \frac{1}{2\pi
i}\int_{\mathbb{R}}\frac{\exp\left(-\zeta^{2\nu}+\tau_t(x)\zeta\right)}{s-x}ds  \\
0 & 1
\end{pmatrix}.
\end{equation}
Then $\Psi(\zeta,\tau_t(x))$ is the unique solution to the following
Riemann-Hilbert problem.
\begin{equation}\label{eq:psi-RHP}
\begin{split}
&1. \quad \text{$\Psi(\zeta,\tau_t(x))$ is analytic on $\mathbb{C}/\mathbb{R}$};\\
&2.\quad \Psi_+(\zeta,\tau_t(x))=\Psi_-(\zeta,\tau_t(x))\begin{pmatrix} 1 & \exp\left(-\zeta^{2\nu}+\tau\zeta\right)  \\
0 & 1
\end{pmatrix},\quad \zeta\in\mathbb{R};\\
&3. \quad \Psi(\zeta,\tau_t(x))=\left(I+O(\zeta^{-1})\right),\quad
\zeta\rightarrow\infty.
\end{split}
\end{equation}
We can use $\Psi(\zeta,\tau_t(x))$ to construct the local parametrix
$S^{x^{\ast}}(x)$. Let us define the matrix $E(x)$ to be
\begin{equation}\label{eq:E-}
E(x)=S^{\infty}(x)e^{Z_t\sigma_3}
\end{equation}
then $E(x)$ is analytic and invertible inside of
$B_{\delta}^{x^{\ast}}$. From (\ref{eq:jumpz-}) and
(\ref{eq:psi-RHP}), we have the following.
\begin{proposition}\label{pro:lpara-}
Let the matrix $S^{x^{\ast}}(x)$ be
\begin{equation}\label{eq:sxast1}
S^{x^{\ast}}(x)=E(x)\Psi(\zeta,\tau_t(x))e^{-Z_t\sigma_3},\quad x\in
B_{\delta}^{x^{\ast}},
\end{equation}
Then, under the double scaling limit (\ref{eq:scalep}),
$S^{x^{\ast}}(x)$ satisfies the conditions
\begin{equation}\label{eq:lRH1}
\begin{split}
&1.\quad \text{$S^{x^{\ast}}(x)$ is analytic in
$B_{\delta}^{x^{\ast}}/\left(B_{\delta}^{x^{\ast}}\cap\mathbb{R}\right)$};\\
&2. \quad S^{x^{\ast}}_+(x)=S^{x^{\ast}}_-(x)J_S(x),\quad x\in
B_{\delta}^{x^{\ast}}\cap\mathbb{R} ;\\
&3. \quad
\text{$S^{x^{\ast}}(x)=\left(I+O\left(n^{-\frac{1}{2\nu}}\right)\right)S^{\infty}(x)$
as $n\rightarrow\infty$, uniformly in $\p B_{\delta}^{x^{\ast}}$}.
\end{split}
\end{equation}
\end{proposition}
\begin{proof}As in the proof of proposition \ref{pro:lpara},
properties 1. and 2. are clear from (\ref{eq:jumpz-}) and
(\ref{eq:psi-RHP}). Let us take a look at the condition at the
boundary of $B_{\delta}^{x^{\ast}}$. We need to be careful as $Z_t$
may contain powers of $n$ in it. At the boundary of
$B_{\delta}^{x^{\ast}}$, we have $\zeta\rightarrow\infty$, and
\begin{equation*}
S^{x^{\ast}}(x)=S^{\infty}(x)\begin{pmatrix} 1 & O\left(e^{2Z_t}\zeta^{-1}\right)  \\
0 & 1
\end{pmatrix}
\end{equation*}
but from the expression of $Z_t$ (\ref{eq:Ztau1}) and the property
of the equilibrium measure (\ref{eq:hineq}), we see that $Z_t<0$ and
hence $e^{2Z_t}$ is bounded. Hence we have
\begin{equation*}
S^{x^{\ast}}(x)=S^{\infty}(x)\left(I+O\left(n^{-\frac{1}{2\nu}}\right)\right)=
\left(I+O\left(n^{-\frac{1}{2\nu}}\right)\right)S^{\infty}(x)
\end{equation*}
This completes the proof of the proposition.
\end{proof}

\subsection{Last transformation of the Riemann-Hilbert problem}

Let us now define $R(x)$ to be the following function.
\begin{equation}\label{eq:Rx}
\begin{split}
R(x)=\left\{
       \begin{array}{ll}
         S(x)\left(S^{r_i}(x)\right)^{-1}, & \hbox{$x$ inside $B_{\delta}^{r_i}$;} \\
         S(x)\left(S^{\infty}(x)\right)^{-1}, & \hbox{$x$ outside of $B_{\delta}^{r_i}$.}
       \end{array}
     \right.
\end{split}
\end{equation}
where $r_1=-2$, $r_2=2$, $r_3=x^{\ast}$ and $S^{\pm 2}(x)$ are the
local parametrices near the edge points $\alpha_t$ and $\beta_t$.
Then the function $R(x)$ has jump discontinuities on the contour
$\Sigma$ shown in Figure \ref{fig:sigma}.

\begin{figure}
\centering
\begin{overpic}[scale=.5,unit=1mm]{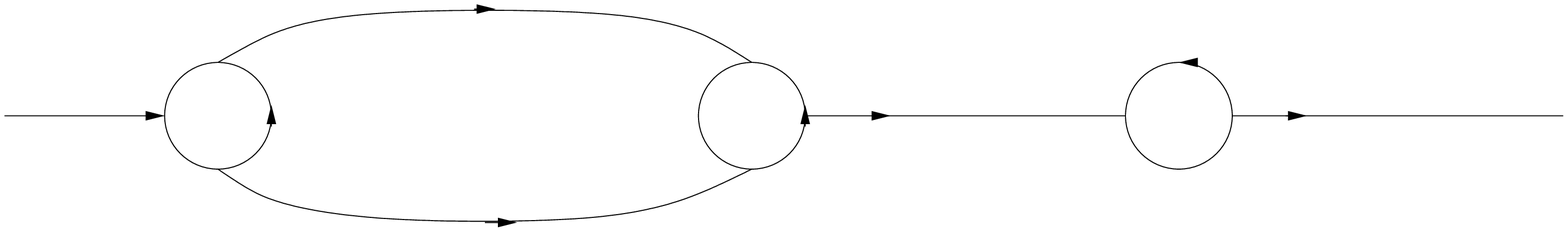}
\end{overpic}
\caption{The contour $\Sigma$ on which $R(x)$ is not
analytic.}\label{fig:sigma}
\end{figure}

In particular, $R(x)$ satisfies the Riemann-Hilbert problem
\begin{equation}\label{eq:RHR}
\begin{split}
&1. \quad \text{$R(x)$ is analytic on $\mathbb{C}/\Sigma$}\\
&2.\quad R_+(x)=R_-(x)J_R(x)\\
&3. \quad R(x)=I+O(x^{-1}),\quad x\rightarrow\infty\\
&4. \quad \text{$R(x)$ is bounded}.
\end{split}
\end{equation}
From the definition of $R(x)$ (\ref{eq:Rx}), it is easy to see that
the jumps $J_R(x)$ has the following order of magnitude.
\begin{equation}\label{eq:Jx}
\begin{split}
J_R(x)=\left\{
       \begin{array}{ll}
         I+O(n^{-1}), & \hbox{$x\in\p B_{\delta}^{-2}\cup B_{\delta}^{2}$ ;} \\
         I+O\left(n^{\frac{2|u^t-\overline{u}^t|-1}{2\nu}}(\log n)^{\frac{u^t}{2\nu}}\right), &
\hbox{$x\in\p B_{\delta}^{x^{\ast}},\quad t>1$;}\\
I+O\left(n^{-\frac{1}{2\nu}}\right), &
\hbox{$x\in\p B_{\delta}^{x^{\ast}},\quad t\leq1$;}\\
I+O\left(e^{-n\gamma}\right), & \hbox{for some fixed $\gamma>0$ on
the rest of $\Sigma$.}
       \end{array}
     \right.
\end{split}
\end{equation}
Since $|u^t-\overline{u}^t|<\frac{1}{2}$, for sufficiently large
$n$, $n^{\frac{2|u^t-\overline{u}^t|-1}{2\nu}}(\log
n)^{\frac{u^t}{2\nu}}$ and $n^{-\frac{1}{2\nu}}$ are small. Then by
the standard theory, \cite{D}, \cite{DKV}, \cite{DKV2}, we have
\begin{equation}\label{eq:Rest}
\begin{split}
R(x)&=I+O\left(n^{\frac{2|u^t-\overline{u}^t|-1}{2\nu}}(\log
n)^{\frac{u^t}{2\nu}}\right),\quad t>1\\
R(x)&=I+O\left(n^{-\frac{1}{2\nu}}\right),\quad t\leq 1.
\end{split}
\end{equation}
uniformly in $\mathbb{C}$.

In particular, the solution $S(x)$ of the Riemann-Hilbert problem
(\ref{eq:RHS}) can be approximated by $S^{\infty}(x)$ and
$S^{r_i}(x)$ as
\begin{equation}\label{eq:approxS}
\begin{split}
S(x)=\left\{
       \begin{array}{ll}
         R(x)S^{r_i}(x), & \hbox{$x$ inside $B_{\delta}^{r_i}$;} \\
         R(x)S^{\infty}(x), & \hbox{$x$ outside of
$B_{\delta}^{r_i}$.}
       \end{array}
     \right.
\end{split}
\end{equation}
When $t>1$, this approximation becomes poor as
$|u^t-\overline{u}^t|$ gets close to $\frac{1}{2}$. However, if we
restrict our attention to a small neighborhood of $x^{\ast}$ such
that
\begin{equation}\label{eq:scale0}
z=\left(x-x^{\ast}\right)n^{\frac{1}{2\nu}}\varphi(x^{\ast}),
\end{equation}
is finite, then we can still use this approximation to obtain the
asymptotic kernel (\ref{eq:critker}).

\section{Asymptotics of the correlation kernel}\label{se:ker}
We should now compute the kernel using the the asymptotics obtained
in section \ref{se:RH}. Recall that the kernel and the solution
$Y(x)$ of the Riemann-Hilbert problem (\ref{eq:RHP}) are related by
(\ref{eq:corker}).

\subsection{Asymptotics of the kernel when $t>1$}

First let us recover the asymptotics of $Y(x)$ from that of $S(x)$.
By reversing the series of transformations (\ref{eq:S}) and
(\ref{eq:Tx}), we find that, for $x\in B_{\delta}^{x^{\ast}}$, the
matrix $S(x)$ and $Y(x)$ are related by
\begin{equation*}
Y(x)=e^{n\frac{\tilde{l}}{2t}\sigma_3}S(x)e^{n\left(g^t(x)-\frac{\tilde{l}}{2t}\right)\sigma_3},\quad
x\in B_{\delta}^{x^{\ast}}.
\end{equation*}
We now use the estimate (\ref{eq:approxS}) and the expression of
(\ref{eq:sxast}) to obtain
\begin{equation*}
Y(x)=e^{n\frac{\tilde{l}}{2t}\sigma_3}R(x)E(x)\Psi^{\nu}(x,\tau_t(x))e^{n\left(g^t(x)-\frac{\tilde{l}}{2t}-u^t\log\zeta-Z_t\right)\sigma_3},\quad
x\in B_{\delta}^{x^{\ast}}.
\end{equation*}
Now from (\ref{eq:jumpz}), we see that the above is equal to
\begin{equation}\label{eq:Ypsi}
Y(x)=e^{n\frac{\tilde{l}}{2t}\sigma_3}R(x)E(x)\Psi^{\nu}(x,\tau_t(x))e^{\left(n\frac{V(x)}{2t}-\frac{\zeta^{2\nu}}{2}+\frac{\tau_t(x)\zeta}{2}\right)\sigma_3}
,\quad x\in B_{\delta}^{x^{\ast}}.
\end{equation}
Let us now study the behavior of $E(x)$ and $R(x)$ in the vicinity
of $x^{\ast}$ when $z$ defined by (\ref{eq:scale0}) is finite. First
let us consider $E(x)$. From (\ref{eq:Ean}), we see that $E(x)$ is
analytic inside the neighborhood $B_{\delta}^{x^{\ast}}$. Then from
the power series expansion of $E(x)$ inside $B_{\delta}^{x^{\ast}}$,
we obtain
\begin{equation}\label{eq:Eord}
E(x)=\left(E^0+E^1zn^{-\frac{1}{2\nu}}+O\left(n^{-\frac{1}{4\nu}}\right)\right)n^{\frac{u^t-\overline{u}^t}{2\nu}\sigma_3}\left(\log
n\right)^{-\frac{u^t}{2\nu}\sigma_3},
\end{equation}
for some constants $E^0$ and $E^1$ that are bounded as
$n\rightarrow\infty$.

Now consider $R(x)$. Let $m$ be the biggest integer such that
\begin{equation*}
m\left(1-2|u^t-\overline{u}^t|\right)< 1.
\end{equation*}
Then from the Riemann-Hilbert problem (\ref{eq:Rx}), we see that
$R(x)$ is analytic inside $B_{\delta}^{x^{\ast}}$, hence in terms of
$z$, we have the following estimate
\begin{equation}\label{eq:Rord}
\begin{split}
R(x)&=I+\sum_{j=1}^m\Lambda_j\left(n^{\frac{2|u^t-\overline{u}^t|-1}{2\nu}}(\log
n)^{\frac{u^t}{2\nu}}\right)^j\\
&+ O\left(zn^{-\frac{1}{2\nu}}(\log
n)^{\frac{u^t}{2\nu}}\right)+O\left(n^{-\frac{1}{2\nu}}\right),
\end{split}
\end{equation}
where $O\left(zn^{-\frac{1}{2\nu}}(\log
n)^{\frac{u^t}{2\nu}}\right)$ denotes $z$ dependent terms that are
of order $n^{-\frac{1}{2\nu}}(\log n)^{\frac{u^t}{2\nu}}$. The
constants $\Lambda_j$ are finite in the limit $n\rightarrow\infty$.

In particular, from (\ref{eq:Eord}) we see that
$E^{-1}(x^{\prime})E(x)$ satisfies the following
\begin{equation}\label{eq:Eord2}
E^{-1}(x^{\prime})E(x)=n^{-\frac{u^t-\overline{u}^t}{2\nu}\sigma_3}\left(\log
n\right)^{\frac{u^t}{2\nu}\sigma_3}\left(I+O\left(\frac{z^{\prime}-z}{n^{\frac{1}{2\nu}}}\right)\right)n^{\frac{u^t-\overline{u}^t}{2\nu}\sigma_3}\left(\log
n\right)^{-\frac{u^t}{2\nu}\sigma_3},
\end{equation}
while $R^{-1}(x^{\prime})R(x)$ satisfies the following estimate
\begin{equation}\label{eq:Rord2}
R^{-1}(x^{\prime})R(x)=I+O\left((z^{\prime}-z)n^{-\frac{1}{2\nu}}(\log
n)^{\frac{u^t}{\nu}}\right)
\end{equation}
Hence the product $E^{-1}(x^{\prime})R^{-1}(x^{\prime})R(x)E(x)$
satisfies the following estimate
\begin{equation}\label{eq:ERord}
E^{-1}(x^{\prime})R^{-1}(x^{\prime})R(x)E(x)=I+O\left((z^{\prime}-z)n^{\frac{2|u^t-\overline{u}^t|-1}{2\nu}}(\log
n)^{\frac{u^t}{\nu}}\right)
\end{equation}
If we now substitute (\ref{eq:Ypsi}) and (\ref{eq:ERord}) back to
(\ref{eq:corker}), we obtain the following estimate for the kernel
\begin{equation}\label{eq:kest}
\begin{split}
K_{n,N}(x,x^{\prime})&=\frac{\varphi(x^{\ast})n^{\frac{1}{2\nu}}}{2\pi
i(z-z^{\prime})}\left(0\quad 1\right)\left(\Psi^{\nu}_+(\zeta^{\prime},\tau_t(x^{\prime}))\right)^{-1}\Psi^{\nu}_+(\zeta,\tau_t(x))\begin{pmatrix} 1 \\
0
\end{pmatrix}\\
&\times\exp\left(\frac{\tilde{V}(\zeta)}{2}+\frac{\tilde{V}(\zeta^{\prime})}{2}\right)\left(1+O\left(n^{\frac{2|u^t-\overline{u}^t|-1}{2\nu}}(\log
n)^{\frac{u^t}{\nu}}\right)\right)
\end{split}
\end{equation}
where $\tilde{V}(\zeta)=-\zeta^{2\nu}+\tau_t(x)\zeta$.

Now recall that from (\ref{eq:etauord}) and (\ref{eq:varphi}), we
have
\begin{equation*}
\lim_{n\rightarrow\infty}\zeta=z,\quad
\lim_{n\rightarrow\infty}\tau_t(x)=0
\end{equation*}
If we now take these into account and substitute (\ref{eq:psi}) into
(\ref{eq:kest}), then we obtain the limit of the kernel as
\begin{equation}\label{eq:klim}
\begin{split}
\lim_{n,N\rightarrow\infty}\frac{1}{\varphi(x^{\ast})n^{\frac{1}{2\nu}}}K_{n,N}(x,x^{\prime})
&=\kappa_{\overline{u}-1}^{\nu}e^{-\frac{(z^{\prime})^{2\nu}+z^{2\nu}}{2}}\frac{\pi_{\overline{u}}^{\nu}(z^{\prime})\pi_{\overline{u}-1}^{\nu}(z)
-\pi_{\overline{u}}^{\nu}(z)\pi_{\overline{u}-1}^{\nu}(z^{\prime})}{2\pi
i(z-z^{\prime})}\\
&= K_{\overline{u}}^\nu(z,z^{\prime}),
\end{split}
\end{equation}
where $u=\lim_{n\rightarrow\infty}u^t$ and
$\kappa_{\overline{u}}^{\nu}=\kappa_{\overline{u}}^{\nu}(0)$.

To complete the proof of the first part of theorem \ref{thm:main},
we need to show that
$\varphi(x^{\ast})=\left(\frac{Q(x^{\ast})\sqrt{(x^{\ast})^2-4}}{2\nu}\right)^{\frac{1}{2\nu}}$.
This can be seen from the expression (\ref{eq:curve}). From
(\ref{eq:curve}), we have
\begin{equation*}
\sqrt{q(x)}=\int_{\mathbb{R}}\frac{\rho(y)}{y-x}dy+\frac{V^{\prime}(x)}{2}=-h^{\prime}(x)+\frac{V^{\prime}(x)}{2},
\end{equation*}
then from the fact that $2h(x^{\ast})-V(x^{\ast})-l=0$ and the
expressions of of $q(x)$ (\ref{eq:q}) and $\zeta$ (\ref{eq:zeta}),
we see that, upon integration, we have
\begin{equation*}
n(x-x^{\ast})^{2\nu}\varphi^{2\nu}(x^{\ast})=n(x-x^{\ast})^{2\nu}\frac{Q(x^{\ast})\sqrt{(x^{\ast})^2-4}}{2\nu}
\end{equation*}
this completes the proof of theorem \ref{thm:main} for the case
$t>1$.

\subsection{Asymptotics of the kernel when $t\leq 1$}

We will now use the local parametrix $S^{x^{\ast}}(x)$
(\ref{eq:sxast1}) constructed for $t\leq 1$ to compute the kernel.
In this case, the solution $Y(x)$ to (\ref{eq:RHP}) is given by
(\ref{eq:Ypsi}) with $\Psi^{\nu}(x,\tau_t(x))$ replaced by
$\Psi(x,\tau_t(x))$ and $\tau_t(x)$ defined by (\ref{eq:Ztau1}).

Let $z$ be the variable defined by (\ref{eq:scale0}) and assume that
$z$ is finite. Then by using the power series expansion of $E(x)$
and $R(x)$ inside $B_{\delta}^{x^{\ast}}$, we obtain
\begin{equation}\label{eq:REord1}
\begin{split}
E(x)&=\left(\Pi(x^{\ast})+\Pi^{\prime}(x^{\ast})\frac{z}{\varphi(x^{\ast})n^{\frac{1}{2\nu}}}
+O\left(zn^{-\frac{1}{4\nu}}\right)\right)e^{Z_t\sigma_3},\\
R(x)&=I+R^0n^{-\frac{1}{2\nu}}+O\left(n^{-\frac{1}{4\nu}}\right)+O\left(zn^{-\frac{1}{4\nu}}\right)
\end{split}
\end{equation}
where we have used (\ref{eq:sin2}) to replace $S^{\infty}(x)$ by
$\Pi(x)$ and $O\left(zn^{-\frac{1}{4\nu}}\right)$ denotes $z$
dependent terms with order $n^{-\frac{1}{4\nu}}$. Therefore the
product $E^{-1}(x^{\prime})R(x^{\prime})R(x)E(x)$ is of order
\begin{equation}\label{eq:ordprod2}
E^{-1}(x^{\prime})R(x^{\prime})R(x)E(x)=e^{-Z_t\sigma_3}\left(I+\frac{\Pi^{-1}(x^{\ast})\Pi(x^{\ast})(z-z^{\prime})}{\varphi(x^{\ast})n^{\frac{1}{2\nu}}}
+O\left(\frac{z-z^{\prime}}{n^{\frac{1}{2\nu}}}\right)\right)e^{Z_t\sigma_3}.
\end{equation}
From (\ref{eq:psi-}), one can easily check that
\begin{equation}
\begin{split}
&\left(0\quad
1\right)e^{\left(-N\frac{V(x^{\prime})}{2}-\frac{\tilde{V}(\zeta^{\prime})}{2}\right)\sigma_3}\Psi_+^{-1}(\zeta^{\prime},\tau_t(x^{\prime}))e^{-Z_t\sigma_3}=e^{\left(Z_t+N\frac{V(x^{\prime})}{2}+\frac{\tilde{V}(\zeta^{\prime})}{2}\right)}\left(0\quad
1\right)
,\\
&e^{Z_t\sigma_3}\Psi_+(\zeta,\tau_t(x))e^{\left(N\frac{V(x)}{2}+\frac{\tilde{V}(\zeta)}{2}\right)\sigma_3}\begin{pmatrix} 1 \\
0
\end{pmatrix}=e^{\left(Z_t+N\frac{V(x)}{2}+\frac{\tilde{V}(\zeta)}{2}\right)}\begin{pmatrix} 1 \\
0
\end{pmatrix}.
\end{split}
\end{equation}
where $\tilde{V}(\zeta)=-\zeta^{2\nu}+\tau_t(x)\zeta$.

We then substitute (\ref{eq:ordprod2}) and (\ref{eq:psi-}) into
(\ref{eq:corker}) and arrive at
\begin{equation*}
\begin{split}
K_{n,N}(x,x^{\prime})&=\frac{e^{\left(2Z_t+\frac{\tilde{V}(\zeta)}{2}+\frac{\tilde{V}(\zeta^{\prime})}{2}\right)}}{2\pi
i}\left(0\quad 1\right)\Pi^{-1}(x^{\ast})\Pi(x^{\ast})\begin{pmatrix} 1 \\
0
\end{pmatrix}\left(1+O\left(n^{-\frac{1}{2\nu}}\right)\right).
\end{split}
\end{equation*}
This gives the double scaling limit of the kernel
\begin{equation*}
\lim_{n,N\rightarrow\infty}e^{-2Z_{t}}K_{n,N}(x,x^{\prime})=e^{-\frac{z^{2\nu}+(z^{\prime})^{2\nu}}{2}}\frac{1}{8\pi}
\left(\frac{1}{x^{\ast}-\beta_t}-\frac{1}{x^{\ast}-\alpha_t}\right).
\end{equation*}
This completes the proof of theorem \ref{thm:main}.

\vspace{.25cm}

\noindent\rule{16.2cm}{.5pt}

\vspace{.25cm}

{\small

\noindent {\sl School of Mathematics \\
                       University of Bristol\\
                       Bristol BS8 1TW, UK  \\
                       Email: {\tt m.mo@bristol.ac.uk}

                       \vspace{.25cm}

                       \noindent  20 November  2007}}

\end{document}